\documentclass[12pt]{article} 
\usepackage[sectionbib]{natbib}
\usepackage{array,epsfig,fancyhdr,rotating}
\usepackage[]{hyperref}  
\usepackage{sectsty, secdot}
\sectionfont{\fontsize{12}{14pt plus.8pt minus .6pt}\selectfont}
\usepackage{amsmath}
\numberwithin{equation}{section}
\subsectionfont{\fontsize{12}{14pt plus.8pt minus .6pt}\selectfont}

\usepackage[margin=1.1in]{geometry}

\usepackage{amssymb}
\usepackage{amsfonts}
\usepackage{multirow}
\usepackage{amsthm}
\usepackage{bm}
\usepackage{amsfonts}
\usepackage{autonum}
\usepackage{graphicx}
\usepackage[table]{xcolor}
\usepackage{array}
\usepackage{rotating}
\usepackage{caption}
\usepackage{changepage} 
\usepackage{booktabs}
\usepackage{subfigure}

\newtheorem{theorem}{Theorem}
\newtheorem{lemma}{Lemma}

\theoremstyle{definition}

\newtheorem{remark}{Remark}

\newcommand{\s}{\textbf{s}}

\newcommand{\h}{\textbf{h}}

\newcommand{\tB}{\textbf{t}}
\newcommand{\T}{\bm{T}}

\newcommand{\YB}{\bm{Y}}

\newcommand{\n}{\mathbb{N}}
\newcommand{\nn}{\mathbf{n}}

\newcommand{\hope}[1]{\mathbb{E}\left[#1\right]}
\newcommand{\hopehat}[1]{\hat{\mathbb{E}}[#1]}
\newcommand{\rr}{\mathbb{R}}
\newcommand{\dsumi}{\displaystyle\sum_{i=1}^{\hat{n}}}

\newcommand{\dsum}[3]{\displaystyle\sum_{#1=#2}^{#3}}
\newcommand{\defi}{\;\dot{=}\;}

\newcommand{\SigmaB}{\boldsymbol{\Sigma}}

\newcommand{\betaB}{\bm{\beta}}

\newcommand{\Var}[1]{\text{Var}\left(#1\right)}

\newcommand{\Xbf}{\textbf{X}}

\renewcommand{\thefootnote}{\fnsymbol{footnote}}

\begin{document}

\renewcommand{\baselinestretch}{1.1}



\renewcommand{\thefootnote}{}
$\ $\par


\fontsize{12}{14pt plus.8pt minus .6pt}\selectfont \vspace{0.8pc}
\centerline{\large\bf A SEMIPARAMETRIC AUTORREGRESIVE SPATIAL }
\vspace{2pt} 
\centerline{\large\bf PREDICTION MODEL }
\vspace{.4cm} 
\centerline{Rodrigo García Arancibia$^{1}$, Pamela Llop$^{2}$, Mariel Guadalupe Lovatto$^{2,}$\footnote{Corresponding author. E-mail: marielguadalupelovatto@gmail.com }} 
\vspace{.4cm} 
\noindent\centerline{\small\it $^{1}$Instituto de Economía Aplicada Litoral (IECAL-FCE-UNL),
		 \& CONICET.}
\noindent\centerline{		
		\small\it $^{2}$Facultad de Ingeniería Química, 
		 (FIQ-UNL) \& CONICET.}
\noindent\centerline{\small\it 	
		  Santa Fe, Argentina.}
\vspace{.45cm} \fontsize{9}{11.5pt plus.8pt minus.6pt}\selectfont




%
%
%
%
%


\begin{quotation}
\noindent {\it Abstract:}
{In this paper we propose a semiparametric spatial autoregressive model that combines a linear covariate component with a nonparametrically estimated spatial term, allowing flexible dependence modeling without restrictive covariance structure while preserving interpretability. We establish asymptotic properties, including consistency and asymptotic normality, and evaluate performance through simulations and real data. Results show competitive predictive accuracy relative to geostatistical methods and improved interpretability compared to spatial econometric models.}\\

\vspace{9pt}
\par
\end{quotation}\par

\def\thefigure{\arabic{figure}}
\def\thetable{\arabic{table}}

\renewcommand{\theequation}{\thesection.\arabic{equation}}

\fontsize{11}{14pt plus.8pt minus .6pt}\selectfont

\section{Introduction}
Statistical models are often developed under the assumption of independence among observations, which simplifies both model specification and theoretical analysis. In practice, however, dependent structures frequently arise, and spatial data offer a good example, as measurements collected at proximate locations often exhibit correlation due to spatial proximity. Ignoring this dependence may lead to biased inference and reduced predictive accuracy making spatial prediction a central problem in a wide range of fields \citep{kiani2025spatial, cerqueti2025spatially}.

Several methodological frameworks have been proposed to address spatial dependence. In geostatistics, kriging-based approaches explicitly model spatial covariance structures and are widely used for spatial prediction \citep{meng2024kriging,khan2023spatial}. More recently, hybrid methods combining geostatistical techniques with machine learning have been proposed to capture nonlinear relationships while accounting for spatial dependence \citep{Fouedjio2024_GeoML}, with applications in areas such as ecology and agriculture \citep{Shen2024_HybridML_Geo}, spatiotemporal analysis \citep{JeongKoo2025_ISPRS}, and remote sensing \citep{Tadic2024_RemoteSensing_SIF}. In contrast, spatial econometric models focus on the estimation and interpretation of covariate effects by introducing spatial autoregressive structures \citep{vagnini2025regional}. However, these models rely on parametric specifications and on an exogenously defined spatial weights matrix, whose choice may substantially affect the results \citep{duncan2017spatial, bauman2018optimizing}.

Motivated by the compromise between predictive flexibility and interpretability, in a previous work we introduced a nonparametric kriging framework in which kriging weights are estimated using kernel smoothing techniques, avoiding the explicit specification of a parametric covariance function or spatial weights matrix \citep{Lovatto:thesis2022,GarciaArancibia2023NPS}. Building on this approach, we propose a semiparametric spatial autoregressive model for spatial prediction that combines a linear component associated with covariates, allowing direct interpretation of marginal effects, with a spatial autoregressive term estimated through the nonparametric kriging approach.

Along these lines, \cite{montero2012sar,minguez2020alternative} in the spatial econometrics literature and \cite{Jenish2014} for censored spatial data, consider semiparametric autoregressive models for spatial data but only account for some covariates nonparametrically. On the other hand, \cite{Gao2006}, from a nonparametric regression perspective, considers semiparametric autoregressive models for spatial data in which the autoregressive component, given by neighbouring responses, can enter both the parametric and the nonparametric parts of the model; the nonlinear component is further specified through an additive structure and estimated via marginal integration combined with local kernel methods.

{In the context of time series and prediction,  semiparametric autoregressive models are also considered in \cite{Aneiros-Perez2008Nonparametric}. They showed asymptotic normality and find convergence rates for the estimator of the linear parameter and for the nonparametric component based on the predecessors (of the time they wished to predict). For their theory, they assumed the response to take values in a semi-metric space $\mathcal{H}$ with semi-metric $d_{\mathcal{H}}(\cdot,\cdot)$. In a related direction, but in the framework of the spatial data, \cite{DaboNiang2016} proposed a kernel-based nonparametric predictor for spatial data using separate kernels for spatial locations and predictors and established uniform consistency and asymptotic normality under $\alpha$-mixing. Unlike the approach considered here, their model corresponds to a classical nonparametric regression framework and does not include spatial autoregressive terms.} We have found these papers useful and combining their approaches allows us to develop our theory. 

The main contributions of this paper can be summarized in three points. First, we introduce a semiparametric spatial autoregressive framework that avoids the explicit parametric specification of the spatial dependence structure. Second, we develop the theoretical properties of the resulting estimators, establishing asymptotic normality and convergence rates for the parametric component together with convergence results for the nonparametric spatial term. These results contribute to the theoretical understanding of estimation in spatial semiparametric models. Third, we evaluate the empirical performance of the proposed method through simulation studies and real data applications, comparing its predictive accuracy with that of widely used spatial autoregressive models.

The remainder of the paper is organized as follows. Section 2 introduces the proposed semiparametric spatial autoregressive model and describes the estimation procedure. Section 3 presents the main theoretical results, including asymptotic properties of the estimators. Section 4 reports results from simulation studies and illustrates the proposed approach using real data applications.

\section{Model and estimators}

This section presents the semiparametric spatial autoregressive model and the corresponding estimation framework. We first describe the model and then introduce the notation and define the estimators.

\subsection{Semiparametric spatial autoregressive model}
The variables of interest are observed at sites \(\s = (i_1, \dots, i_d) \in \mathbb{Z}^d\), which belong to a grid defined by
\[
\mathcal{D}_{\nn} = \left\{ \s : 1 \le i_l \le n_l,\; l = 1, \ldots, d \right\},
\]
where \(n_l\) denotes the number of points along the \(l\)-th direction, \(\textbf{n} = (n_1, \ldots, n_d)\), and we define \(\hat{n} \defi n_1 \times \cdots \times n_d\) as the sample size. For simplicity, we assume that \(n = n_1 = \cdots = n_d\) and that \(\hat{n} \to \infty\) as \(n \to \infty\). If the point where we want to predict \(\s_0 \in \mathcal{D}_{\nn}\) is \(\mathcal{O}_{\textbf{n}} = \mathcal{D}_{\nn} \setminus \{\s_0\}\), the grid of observed points; otherwise \(\mathcal{O}_{\textbf{n}} = \mathcal{D}_{\nn}\).

Let $\{(Y_i, \bm{X}_{i}^T, \bm{T}_i^T)\}_{\s_i \in \mathcal{O}_{\nn} }$ be a sample of random vectors defined on $(\Omega, \mathcal{F}, \mathbb{P})$, where $\mathcal{O}_{\nn} \subset \mathbb{Z}^d$. Define $ Y_i\doteq  Y(\s_i)\in \rr$, $ \bm{X}_{i}\doteq \bm{X}(\s_i)\in \rr^p$, and $\bm{T}_{i} \doteq \bm{T}(\s_i)  \in \rr^k$. Moreover, 
\(\label{Ti}
\bm{T}_i =(Y_{i(1)}, \dots, Y_{i(k)})^T,
\)
where $\{\s_{i(1)}, \dots,\s_{i(k)}\}$ denote the $k$ nearest neighbors of site $\s_i$. We call $\bm{T}_i$ the neighborhood vector, that is, the vector of values of the response variable measured at the $k$ nearest neighbors of $\s_i$.  To simplify notation, we henceforth omit the dependence on \(\nn\) and write \(\{(Y_i, \bm{X}_i^T, \bm{T}_i^T)\}\) for the sample of random vectors observed at sites in \(\mathcal{O}_{\nn}\).

In this context, the proposed semiparametric autoregressive model is given by 
\begin{equation}\label{modeloespacial}
	Y_i = \bm{X}^T_{i}\betaB + r(\bm{T}_i) + \varepsilon_i, 
	\qquad \forall \; \s_i \in  \mathcal{D}_{\textbf{n}},
\end{equation}
where $\bm{\beta} = (\beta_1, \dots, \beta_p)^T \in \rr^p$, $r : \rr^k \to \rr$ is an unknown smooth function, and the errors satisfy
\[\label{modeloespacialhope}
	\hope{\varepsilon_i \mid \bm{X}_i, \bm{T}_i} = 0.
\]
Model \eqref{modeloespacial} combines a parametric component associated with the covariates and a nonparametric component capturing spatial dependence through the neighbourhood vector.

\subsection{Notation and definition of estimators}
The conditional expectation is
\begin{align}
	\hope{Y_i \mid \bm{T}_{i}}  
    \label{poblacional}
	&= \hope{\bm{X}^T_{i} \mid \bm{T}_{i}} \bm{\beta} + r(\bm{T}_{i}),
\end{align}
since $\hope{\varepsilon_i \mid \bm{X}_i, \bm{T}_i} = 0$ implies $\hope{\varepsilon_i \mid \bm{T}_{i}} = 0$. Following \cite{robinson1988}, subtracting \eqref{poblacional} from \eqref{modeloespacial} yields
\begin{equation}\label{modeloresta}
	Y_i - \hope{Y_i \mid \bm{T}_{i}}  
	= (\bm{X}_i - \hope{\bm{X}_{i} \mid \bm{T}_{i}})^T \bm{\beta} + \varepsilon_i.
\end{equation}

Let $\tB \defi \T(\s)$. For fixed $\h$, define
\begin{equation}\label{hopeY}
	\hopehat{Y_i \mid \bm{T}_{i} = \tB} 
	= \sum_{\substack{\s_k \in \mathcal{O}_\textbf{n} \\ \s_k \neq \s }}
	\omega_{\h}(\tB, \bm{T}_{k}) \, Y_{k},
\end{equation}
and
\begin{equation}\label{hopeX}
	\hopehat{X_{ij} \mid \bm{T}_{i} = \tB} 
	= \sum_{\substack{\s_k \in \mathcal{O}_\textbf{n} \\ \s_k \neq \s }}
	\omega_{\h}(\tB, \bm{T}_{k}) \, X_{kj},
\end{equation}
The Nadaraya–Watson kernel regression function estimator.


Define the matrix $\bm{W}_{\h} \in \rr^{\hat{n}\times \hat{n}}$ with entries $(\bm{W}_{\h})_{ij} = \omega_{\h}(\bm{T}_i,\bm{T}_{j})$. Then \eqref{hopeY} can be written in matrix form as
\begin{equation}\label{hopeYM}
	\hopehat{{Y}_i \mid \bm{T}_{i}} = (\bm{W}_{\h} \bm{Y})_{i},
\end{equation}
where $(\bm{W}_{\h} \bm{Y})_{i}$ denotes the $i$-th component of $\bm{W}_{\h} \bm{Y} \in \rr^{\hat{n}}$. Similarly, \eqref{hopeX} becomes
\begin{equation}\label{hopeXM}
	\hopehat{X_{ij} \mid \bm{T}_{i}} = (\bm{W}_{\h}\mathbb{X})_{ij},
\end{equation}
where $(\bm{W}_{\h}\mathbb{X})_{ij}$ denotes the $(i,j)$-th entry of $\bm{W}_{\h}\mathbb{X} \in \rr^{\hat{n} \times p}$ and ${\mathbb{X}} \doteq ({\bm{X}}_1, \dots, {\bm{X}}_{\hat{n}})^T \in \rr^{\hat{n}\times p}$ the matix of covariates.

Based on these estimators, define
$
	\tilde{Y}_i \doteq Y_i - \hopehat{Y_i \mid \bm{T}_{i}} = Y_i - (\bm{W}_{\h} \bm{Y})_{i},
$
and, letting $(\bm{W}_{\h}\mathbb{X})_{i\cdot}$ denote the $i$-th row of $\bm{W}_{\h}\mathbb{X}$,
\begin{equation}\label{xp}
	\tilde{\bm{X}}_i \doteq \bm{X}_i - \hopehat{\bm{X}_{i} \mid \bm{T}_{i}} 
	= \bm{X}_i - (\bm{W}_{\h}\mathbb{X})_{i\cdot}.
\end{equation}
Hence, from \eqref{modeloresta},
\[\label{yp}
	\tilde{Y}_i = \tilde{\bm{X}}_i^T \bm{\beta} + \varepsilon_i.
\]
If $\tilde{\bm{Y}} \doteq (\tilde{Y}_1, \dots, \tilde{Y}_{\hat{n}})^T \in \rr^{\hat{n}}$ and 
$\tilde{\mathbb{X}} \doteq (\tilde{\bm{X}}_1, \dots, \tilde{\bm{X}}_{\hat{n}})^T \in \rr^{\hat{n}\times p}$ and if $\tilde{\mathbb{X}}^T\tilde{\mathbb{X}}$ is invertible, the least squares estimator of $\bm{\beta}$ is
\begin{equation}\label{betahat}
	\hat{\bm{\beta}}_{\h} = (\tilde{\mathbb{X}}^T\tilde{\mathbb{X}})^{-1} \tilde{\mathbb{X}}^T \tilde{\bm{Y}}.
\end{equation}
\begin{remark}
The estimation procedure removes from both the response and the covariates the component explained by $\bm{T}_i$. Consequently, the coefficients $\bm{\beta}$ capture the linear relationship between $\bm{X}_i$ and $Y_i$ that is not explained by the local spatial information. In this sense, $\bm{\beta}$ measures the effect of $\bm{X}_i$ on $Y_i$ after accounting for the spatial dependence induced by observations at neighboring locations.
\end{remark}
We next estimate the function $r$. From Equation \eqref{poblacional} we have that, 
\begin{equation}\label{defrhat}
	\hat{r}_\h(\bm{T}_i) 
	= \hopehat{Y_i \mid \bm{T}_{i}}  
	- \hopehat{\bm{X}_{i} \mid \bm{T}_{i}} \, \hat{\bm{\beta}}_\h.
\end{equation}
Substituting \eqref{hopeYM}, \eqref{hopeXM}, and \eqref{betahat} into \eqref{defrhat}, we obtain
\begin{align}
	\hat{r}_\h(\bm{T}_i) 
	&= (\bm{W}_{\h} \bm{Y})_{i} 
	- (\bm{W}_{\h}\mathbb{X})_{i\cdot} \, \hat{\bm{\beta}}_\h \notag \\
	&= (\bm{W}_{\h})_{i\cdot} \bm{Y} 
	- (\bm{W}_{\h})_{i\cdot} \mathbb{X} \, \hat{\bm{\beta}}_\h \notag \\ \label{mhat}
	&= (\bm{W}_{\h})_{i\cdot} \bigl(\bm{Y} - \mathbb{X}\hat{\bm{\beta}}_\h \bigr).
\end{align}

Finally, under the proposed model, the prediction of the response $Y_0$ at location $\s_0 \in \rr^d$ is given by
\begin{align}
	\hat{Y}_0 
	&= \hopehat{Y_0 \mid \bm{X}_{0},\bm{T}_{0}} \notag \\
	&= \bm{X}_{0}^T \hat{\bm{\beta}}_\h + \hat{r}(\bm{T}_{0}) \label{BLUPnp}\\
	&= \bm{X}_{0}^T \hat{\bm{\beta}}_\h 
	+ (\bm{W}_{\h})_{0\cdot} \bigl(\bm{Y} - \mathbb{X}\hat{\bm{\beta}}_\h \bigr).
\end{align}
Here, $\bm{X}_{0} \in \mathbb{R}^{p}$ denotes the vector of covariates observed at $\s_0$, and $\bm{T}_0 \in \mathbb{R}^{k}$ denotes the neighborhood vector associated with $\s_0$.

To estimate the parameter vector $\h=\{h_{1\nn}, h_{2\nn}, k\}$, that is, the smoothing parameters and the number of nearest neighbours $k$, we employ a $k$-fold cross-validation procedure, which is described in detail in the following section.

\section{Assumptions and main results}
To estimate the conditional expectations \eqref{hopeY} and \eqref{hopeX}, we propose Nadaraya--Watson type estimators,  constructed as in \cite{GarciaArancibia2023NPS}. In particular,
\begin{equation}\label{pesos0}
	\omega_{\h}(\tB,\bm{T}_{k}) \doteq
	\frac{K_{1, h_{1\nn}}\big( d(\s, \s_k) \big) 
	K_2\!\left(\frac{d_m(\tB,\bm{T}_k)}{h_{2\nn}}\right)}
	{\sum_{\substack{\s_j \in \mathcal{O}_\textbf{n} \\ \s_j \neq \s }}
	K_{1, h_{1\nn}}\big( d(\s, \s_j) \big)
	K_2\!\left(\frac{d_m(\tB,\bm{T}_j)}{h_{2\nn}}\right)}.
\end{equation}
Here, $d$ is the Euclidean distance,  $
		d_{m}(\tB, \bm{T}_k) \doteq |Med(\tB)-Med(\bm{T}_i)|,
	$ is a similarity measure where $Med(\tB)$ the median of $\tB$. In adittion $\h = \{h_{1\nn}, h_{2\nn}, k\}$ and we also set $\omega_{\h}(\tB,\bm{T}_{k}) = 0$ whenever $\s = \s_k$. Moreover \(K_1\) and \(K_2\) are kernels defined on \(\mathbb{R}\), and \(h_{1\nn}\), \(h_{2\nn}\) are bandwidth sequences tending to zero, such that
\(
\hat{n}\, h_{1\nn}^d\, h_{2\nn}^k \to \infty.
\)
\begin{remark}
    In contrast to \cite{Gao2006} approaches that mitigate the curse of dimensionality by imposing an additive structure on the neighborhood effect, the proposed methodology avoids such separability assumptions by relying instead on a dimension-reduction strategy based on summary measures of the neighborhood vector. Moreover, we added a kernel function for euclidean distance like \cite{DaboNiang2016}, to balance the loss of spatial structure when using the median. 
\end{remark}
Thus, we define
\[
K_{1, h_{1\nn}} \big( d(\s, \s_k) \big)
=
K_1\!\left(\frac{d\!\left( \frac{\s}{\nn}, \frac{\s_k}{\nn}\right)}{h_{1\nn}} \right),
\qquad 
\frac{\s}{\nn} = \left(\frac{i_1}{n}, \ldots, \frac{i_d}{n}\right).
\]

For each fixed site \(\s\), define the neighborhood size as
\[
k_\nn = \sum_{\s_i \in \mathcal{D}_\nn} \mathbb{I}_{\{ d(\s, \s_i) \le r_\nn \}},
\]
where \(r_\nn > 0\) is a growing neighborhood radius. The quantity \(k_\nn\) represents the number of locations involved in the local prediction around \(\s\).

Then, the condition
\(
{d\!\left( \frac{\s}{\nn}, \frac{\s_i}{\nn} \right)}/{h_{1\nn}} \le 1,
\)
is equivalent to \(d(\s, \s_i) \le r_\nn\), with
\(
r_\nn = n\, h_{1\nn}.
\)
We assume that
\(
r_\nn \to \infty,
\) so that
\((n h_{1\nn})^d = \hat{n} h_{1\nn}^d \to \infty,
\)
ensuring that the effective number of observations used in the local prediction diverges.

From a geometric viewpoint, on a regular lattice in \(\mathbb{Z}^d\), the number of points inside a Euclidean ball \(B(r)\) of radius \(r\) satisfies $
\#\big(\mathbb{Z}^d \cap B(r)\big) = \operatorname{Vol}_d(B(r)) + O(r^{d-1}), $ which implies
\(
k_\nn = O(r_\nn^d).
\)
Furthermore, we assume the asymptotic expansion
\begin{equation}\label{expansionk_n}
k_\nn = C_d\, r_\nn^d + O(r_\nn^\beta),
\qquad 0 < \beta < d,
\end{equation}
where \(C_d\) is the volume of the unit ball in \(\mathbb{R}^d\).

These conditions capture the bias--variance trade-off: \(k_\nn \to \infty\) reduces variance, while \(h_{1\nn} \to 0\) controls the spatial smoothing bias.

\subsection{Technical assumptions}
The results of the proposed predictor are achieved under the following assumptions
on the regression and densities, the kernels, bandwidths and local dependence condition. Let $\bm{T} \in \mathcal{C}$, where $\mathcal{C}$ is a compact subset of $\mathbb{R}^k$ and

\begin{enumerate}
    \item \label{K} The kernel functions $K_1$ and $K_2$ have support on $[0,1]$ and are Lipschitz continuous on $[0,\infty)$. Moreover, there exist constants $k_i>0$, $i=1,2$, such that for all $u\in[0,1]$,
\(
-K_i'(u)>k_i>0,
\)
and $K_i(u)>0$.

\item \label{fs} The marginal density $f$ of $\bm{T}$ with respect to the Lebesgue measure on $\rr^k$ exists, is Lipschitz, and satisfies
\(
\inf_{\tB\in \mathcal{C}} f(\tB) \ge \delta > 0.
\)
The joint density $f_{\bm{T}_i,\bm{T}_j}$ of $(\bm{T}_i,\bm{T}_j)$ with respect to the Lebesgue measure on $\rr^{2k}$ exists and is bounded. Moreover, there exists a constant $C>0$ such that
\(\big|f_{\bm{T}_i,\bm{T}_j}(\tB_1,\tB_2)-f_{\bm{T}_i}(\tB_1)f_{\bm{T}_j}(\tB_2)\big|<C,
\)
for all $\tB_1,\tB_2$ and $\s_i\neq\s_j$.

\item \label{fs2}
\begin{itemize}
\item[(i)] The conditional densities $f_{\T_i, \T_j \mid Y_i, Y_j}$ of $(\T_i,\T_j)$ given $(Y_i,Y_j)$ and 
$f_{\T_i \mid Y_i}$ of $\T_i$ given $Y_i$ exist and are uniformly bounded. That is, there exists a constant $C>0$ such that
\(
f_{\T_i, \T_j \mid Y_i, Y_j}(\tB_1,\tB_2 \mid y_1,y_2) \le C,
\) and
\(f_{\T_i \mid Y_i}(\tB_1 \mid y_1) \le C,
\)
for all $y_1,y_2,\tB_1,\tB_2,\s_i,\s_j$.

\item[(ii)] $\sup_i \hope{|Y_i|^r} < \infty$ and 
$\sup_{\bm{q}} \int |y|^r f_{\T_i,Y_i}(\bm{q},y)\,dy < C$
for some $r>4$ and for all $\s_i$.
Let $g_2(\bm{q}) = \operatorname{Var}(Y_i \mid \T_i=\bm{q})$, which does not depend on $\s_i$ and is continuous in a neighborhood of $\tB$, so that
\(
\sup_{\bm{q}:\|\bm{q}-\tB\|\le h} |g_2(\bm{q})-g_2(\tB)| = o(1)
\quad \text{as } h\to 0.
\)
Let $m(\cdot)$ denote a purely nonparametric regression function (in contrast with the semiparametric regression function considered in our model), and define
\(
g_r(\bm{q}) = \hope{|Y_i - m(\tB)|^r \mid \T_i=\bm{q}},
\)
which is independent of $\s_i$ and continuous in a neighborhood of $\tB$. 
In addition,
\(
g(\bm{q}_1,\bm{q}_2,\tB)
=
\hope{(Y_i-m(\tB))(Y_j-m(\tB)) \mid \T_i=\bm{q}_1,\T_j=\bm{q}_2},
\)
for all $\s_i \ne \s_j$, does not depend on $\s_i,\s_j$ and is continuous in a neighborhood of $(\tB,\tB)$.
\end{itemize}

Analogous conditions (i)–(ii) are assumed to hold for each component of the vector $\bm{X}_i$ in place of $Y_i$.

\item \label{ryg} 
Let $g_j(\tB) = \hope{X_{ij} \mid \T_i=\tB}$ and 
$r(\tB) = \hope{Y_i - \bm{X}_{i}^T\betaB \mid \T_i=\tB}$ 
for $1 \le i \le \hat{n}$ and $1 \le j \le p$. 
Assume that all functions to be estimated are Lipschitz continuous. 
That is, for all $(\tB_1,\tB_2) \in \mathcal{C}\times\mathcal{C}$ and for all $h \in \{r, g_1, \dots, g_p\}$, 
\(
|h(\tB_1)-h(\tB_2)| \le C \|\tB_1-\tB_2\|,
\)
for some constant $C<\infty$.

\item \label{mixing} We assume that \( \{(Y_i, \Xbf^T_i, \bm{T}^T_i)\}_{i=1}^{\hat{n}}\) is jointly strictly stationary and generated by a strongly mixing process whose mixing coefficient \( \alpha(r) \) satisfies $\alpha(r) \leq C r^{-b},
\quad \text{for some } b > 4.5d.$ 
\item \label{epsilonEta} Let \( \bm{\eta}_i \doteq (\eta_{i1}, \dots, \eta_{ip})^T \), where \( \eta_{ij} = X_{ij} - \mathbb{E}(X_{ij} \mid \T_i) \) for \( j = 1, \dots, p \). We assume that, for each \( i = 1, \dots, n \),
\(
\bm{\eta}_i \) is independent of  \(\varepsilon_i .
\)

\item \label{covas} Let
\(
\mathbf{V}_{\bm{\varepsilon}} = \mathbb{E}(\bm{\varepsilon} \bm{\varepsilon}^T)\), {with} \(
\bm{\varepsilon}^T = (\varepsilon_1, \dots, \varepsilon_{\hat{n}})
\) {and} \(
\bm{\eta}^T = (\bm{\eta}_1, \dots, \bm{\eta}_{\hat{n}}) \in \rr^{p \times \hat{n}} .
\)
We assume that
\(
\mathbf{B} = \hope{\bm{\eta}_1 \bm{\eta}_1^T}
\) {and} \(
\mathbf{C} = \lim_{\hat{n} \to \infty} \hat{n}^{-1} \hope{\bm{\eta}^T \mathbf{V}_\varepsilon\bm{\eta}}
\)
are positive definite matrices.

\end{enumerate}
\begin{remark}\label{hipAnne}
   Assumptions H1--H5 of Theorem 3.1 in \cite{DaboNiang2016} are directly ensured by assumptions A2--A5 of this work. In particular, assumption A5, together with the conditions imposed on the smoothing parameters (as stated in the results), allows us to verify assumptions H6 and H7 of that theorem. Moreover, assumption H8, as well as the condition required to apply Remark 4 in the cited work, are guaranteed by the strict stationarity assumption in A6. Finally, the compactness of the set \(\mathcal{C}\) enables us to employ a similarity measure based on differences between medians, thereby preserving the results obtained by the authors.
\end{remark}

\begin{remark}
    All the assumptions are relatively mild in this type of problem and can be justified in detail. For example, assumptions 1 y 2  are quite natural and corresponds to the used in the non-spatial case. Assumptions 3(i) and 3(ii) are required to establish asymptotic properties in the pure nonparametric setting via \cite{DaboNiang2016} and the existence of moments of higher than second order is required for this kind of problem when uniform convergence for nonparametric regression estimation is involved.     
    As shown in Theorem \ref{teo1}   below, condition of independence in Assumption 6, the existence of the limit defining the matrix C, and the positive definiteness of both B and C are required for the formulation of the theorem.
\end{remark}
\subsection{Main results}
We can now state the asymptotic properties of the parametric estimator, including its asymptotic normality and a law of the iterated logarithm.
\begin{theorem}\label{teo1}
Under assumptions 1--7, and additionally, as $\hat{n} \to \infty$, 
\[\label{H0}
\hat{n}h_{2\hat{n}}^{4} \to 0 \qquad \text{ and } \qquad \dfrac{\log \hat{n}}{h_{1\hat{n}}^{d}h_{2\hat{n}}^{k} \hat{n}^{\frac{1}{4}-\frac{1}{r}}} \to 0,
\] it follows that
\begin{equation}\label{normalidadbeta}
\sqrt{\hat{n}}\left(\hat{\boldsymbol{\beta}}_{\mathbf{h}} - \boldsymbol{\beta}\right)
\overset{d}{\longrightarrow} \mathcal{N}(0, \mathbf{A}), 
\quad \text{where } \mathbf{A} = \mathbf{B}^{-1}\mathbf{C}\mathbf{B}^{-1}.
\end{equation}
Moreover,
\begin{equation}\label{convergenciabeta}
\limsup_{\hat{n}\to \infty}
\left(\frac{\hat{n}}{2\log \log \hat{n}}\right)^{\frac{1}{2}}
\left|\hat{\boldsymbol{\beta}}_{\mathbf{h},j}-\beta_j\right|
= a_{jj}^{\frac{1}{2}} \quad \text{a.s.,}
\end{equation}
where $a_{jj} = A_{jj}$.
\end{theorem}

Next we state the result for the nonparametric component.
\begin{theorem}\label{r}
 Under assumptions 1--7, and additionally, as $\hat{n} \to \infty$,
\(
\hat{n}h_{2\hat{n}}^{4} \to 0 \) and \( \frac{\log \hat{n}}{h_{1\hat{n}}^{d}h_{2\hat{n}}^{k} \hat{n}^{\frac{1}{4}-\frac{1}{r}}} \to 0,
\)
we have
\begin{equation}\label{convergenciaM}
\sup_{\boldsymbol{t} \in \mathcal{C}} \left| \hat{r}(\boldsymbol{t}) - r(\boldsymbol{t}) \right|
= O(h_{2\hat{n}}) + O\left( \sqrt{\frac{\log \hat{n}}{\hat{n} \, h_{1\hat{n}}^{d} h_{2\hat{n}}^{k}}} \right)
\quad \text{a.s.}
\end{equation}
\end{theorem}
The proofs of these theorems are given in the Appendix.



\section{Simulation study}

In this section, we investigate the performance of the proposed predictor through a series of simulation studies conducted under controlled settings. The experiments are designed to assess the behavior of the competing predictors across different scenarios, with particular emphasis on the role of the spatial autocorrelation structure and the sampling grid configuration. Predictive accuracy is evaluated using the root mean squared error (RMSE). For each method, we report the mean and standard deviation computed over 50 replications and summarize the empirical distribution of the errors using boxplots.

\subsection{Procedure of prediction of $\hat{Y}_0$ and methodology}

\begin{description}
\item[Step 1.] Specify sets of bandwidths $S(h_{1})$ and $S(h_{2})$ of respectively $K_1$ and $K_2$ and the set of the number of nearest neighbors $k$, $S(k)$.
\item[Step 2.] For each $h_{1\nn} \in S(h_{1})$, $h_{2\nn} \in S(h_{2})$, $k \in S(k)$ and each site $\s_0 \in \mathcal{D}_\nn$, compute  the matrix $\bm{W}_{\h} \in \rr^{\hat{n}\times \hat{n}}$ with entries $(\bm{W}_{\h})_{ij} = \omega_{\h}(\bm{T}_i,\bm{T}_{j})$  defined in \eqref{pesos0}. With that, we compute $\hat{\betaB}_\mathbf{h}$ via \eqref{xp} and \eqref{yp} and $\hat{r}_\h(\bm{T}_0) $ through \eqref{mhat}. Finally, compute $\hat{Y}^{\h}_0=\bm{X}_{0}^T \hat{\bm{\beta}}_\h 
	+ (\bm{W}_{\h})_{0\cdot} \bigl(\bm{Y} - \mathbb{X}\hat{\bm{\beta}}_\h \bigr), $ with $(\bm{W}_{\h})_{0j}=\omega_\h(\bm{T}_0,\bm{T}_{j})$ y $j = 1, \dots, \hat{n}$.

\item[Step 3.] Compute optimal bandwidths $h_{1\nn,\mathrm{opt}}$, $h_{2n,\mathrm{opt}}$ and $k_{\mathrm{opt}}$ by applying a cross-validation procedure over $S(h_1)$, $S(h_2)$ and $S(k)$. More precisely, consider the following minimisation problem, i.e. determine $h_{1\nn,\mathrm{opt}}$, $h_{2n,\mathrm{opt}}$ and $k_{\mathrm{opt}}$ minimising the RMSE over the $\hat{n}$ sites, $
\min_{\h} \sqrt{\frac{1}{n} \sum_{\s_0 \in \mathcal{I}_n}
\left( \hat{Y}_{0}^{\,\h} - Y_{0} \right)^2}. $
\item[Step 4.] For each site $\s_0$, compute $\hat{Y}_{0}^{\,\h_{\mathrm{opt}}}$ corresponding to $h_{1\nn,\mathrm{opt}}$, $h_{2n,\mathrm{opt}}$ and $k_{\mathrm{opt}}$. Thus, this procedure is used in the subsequent analysis, in which we aim to study the behaviour of our predictor. All the following numerical analysis were carried out using the R software (version 4.4.3).
\end{description}

To assess the variability of the errors reported in the next sections, the original sample of size $\hat{n}$ is randomly split into a training set (70\%) and a test set (30\%) \citep{WADOUX2021}. Model parameters, both for the proposed methods and for the competing approaches, are estimated using the training set, while predictive performance is evaluated on the test set. Prediction accuracy is measured by the RMSE. In addition, to evaluate the ability of the methods to estimate the linear parameter, we consider the mean absolute error (MAE) of the parameter estimates, defined as $
MAE_\beta = \frac{1}{p} \sum_{i=1}^{p} \lvert \beta_i - \hat{\beta}_i \rvert.$ 

This procedure is repeated 50 times to obtain 50 prediction errors
\subsection{Settings}
Let \(D = [0,1]^2\). We consider regular, irregular, and clustered (focused) sampling grids (see Figure \ref{fig:grillas}) with sample sizes \(\hat{n} = 100, 625, 1089\). A realization of size \(\hat{n} = 1936\) is generated and nested sub-squares are used to extract subsamples: the 100-point square is contained in the 625-point square, which is in turn contained in the 1089-point square, ensuring all samples arise from the same underlying random field.
\begin{figure}
	\centering
	\includegraphics[width=.9\linewidth]{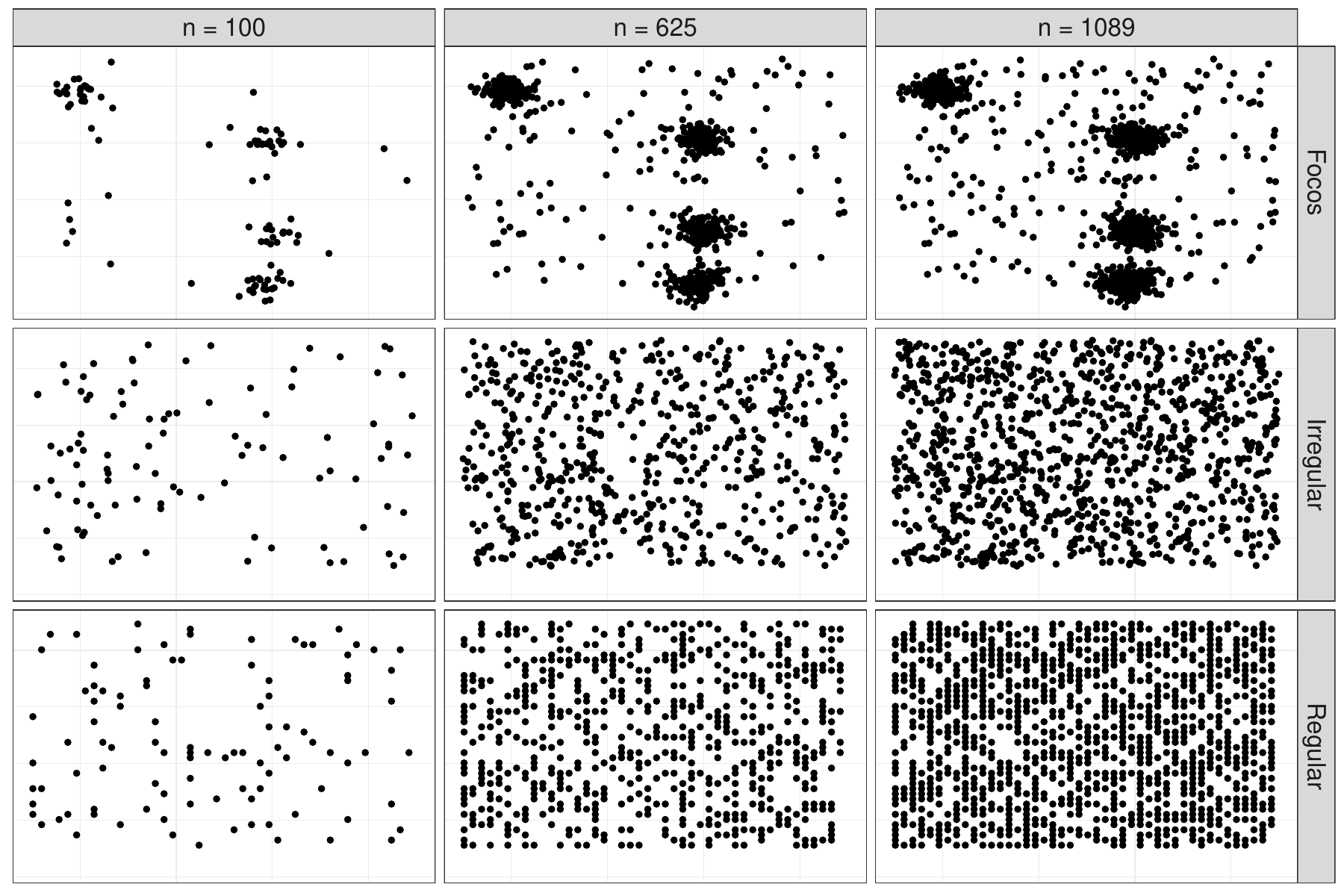}
	\caption{Grid variation as a function of size and structure.}
	\label{fig:grillas}
\end{figure}
Covariates are simulated using the algorithm propose by \cite{Oliver2003}, with \(p = 8\). For each \(\bm{\epsilon}_u\) and \(\bm{\epsilon}_v\), we generate \(\hat{n}\) realizations from a standard normal distribution. If \(u\) is even and \(v = u+1\), the correlation satisfies \(\rho_{uv} \sim U[-1,1]\). Mean functions are spatially constant but differ across covariates, i.e., \(\mu_u(\s_i) = c_u\) with \(c_u \sim U[0.5,2]\) for \(u = 1,\dots,p\). 

Variance--covariance matrices \(\SigmaB_{uu}\) are defined via
\(
(\SigmaB_{uu})_{ij} = \sigma_u^2\, \textbf{R}_{uu}(\s_i,\s_j).
\)
 with \(\sigma_u = \sqrt{\Var{X_u(\s)}}\) constant over \(\s\) and set \(\sigma_u^2 = 1\) for all \(u=1,\dots,p\).
Autocorrelation matrices $\textbf{R}(\s, \tB)$ are specified through a distance-based function, i.e.,
\(
(\bm{R}_{uv})_{ij} = R(\|\s_i - \s_j\|),
\)
where \(R: \mathbb{R}^+ \to \mathbb{R}\). We consider four different choices of \(R\), randomly assigned across the \(p\) covariates.
We consider four autocorrelation models. First, a spherical model with nugget effect \(m_2=0.1\), partial sill \(m_1=1\),  and range \(r=0.5\); and its counterpart without nugget (\(m_2=0\)). Second, a Gaussian model with nugget \(m_2=0.1\), \(m_1=1\), and \(r=0.5\). Finally, a sinc model (a J-Bessel model with \(\alpha=\tfrac{1}{2}\)),
\begin{equation}\label{sinc}
R(h)=
\begin{cases} 
1, & h = 0, \\
\frac{m_1}{m_1+m_2}\left[\left(\frac{2\theta_{21}}{h}\right)^{1/2}\Gamma\left(\tfrac{3}{2}\right)J_{1/2}\left(\frac{h}{\theta_{21}}\right)\right], & h \ne 0,
\end{cases}
\end{equation}
with \(m_1=1\), \(m_2=0.1\), range \(r=0.05\), where \(\Gamma\) is the gamma function and \(J_{1/2}\) is the Bessel function of the first kind of order \(1/2\).

The response vector \(\bm{Y} \) is generated using the simulated covariates and the spatial autoregressive model
\begin{align}\label{econometrico}
\bm{Y} &= \mathbb{X}\betaB + \rho \bm{VY} + \bm{\varepsilon}
= (\bm{I} - \rho \bm{V})^{-1}(\mathbb{X}\betaB + \bm{\varepsilon}),
\end{align}
where \(\bm{V} \in \mathbb{R}^{\hat{n}\times \hat{n}}\) is a neighborhood matrix defined by the the weight defined in \eqref{pesos0}, considering only the kernel associated with the Euclidean distance 
with \(h=0.5\). The error \(\bm{\varepsilon}\) is standard normal noise, and \(\betaB \in \mathbb{R}^p\) has independent components drawn from \(N(0,10^2)\). The parameter \(\rho\) controls the strength of spatial dependence, ranging from no interaction (\(\rho=0\)) to increasing influence of neighboring responses; we consider \(\rho \in \{0, 0.6, 0.9\}\). 

All datasets are standardized by centering and scaling each variable; for the test set, standardization uses only the training-set parameters.

\subsection{Results on prediction}

To organize the results, the analysis is divided into two parts. First, we compare the overall performance of the methods with respect to sample size, grid structure used for data generation, and the parameter $\rho$ (defined in \eqref{econometrico}) using prediction error as the evaluation metric. Second, we evaluate the estimation of the linear parameter.

We examine the empirical distribution of RMSE values using boxplots. Although all figures display the same metric, the scale of the $y$-axis may vary across cases to improve visualization.

Methods are presented in the following order. The first four correspond to semiparametric predictors that differ in the specification of the nonparametric component: KS1 uses only geographical distances; K1ME uses only distances between response medians; K2ME combines both distances through two separate kernels; and K1M incorporates both distances within a single kernel evaluated at their product (see \cite{GarciaArancibia2023NPS}). These are followed by kriging with external drift (KED), lattice kriging (LK), the econometric model used for data generation (SAR), and, finally, ordinary least squares (OLS).
\begin{sidewaysfigure}
			\centering
			\begin{tabular}{>{\columncolor[gray]{0.9}} c c c c}
				\rowcolor[gray]{0.9}&  \multicolumn{3}{c}{\textbf{$n = 100 \hspace{5cm} n = 625 \hspace{5cm} n = 1089$}}\\
				{\begin{sideways}{\hspace{1.5cm}$\rho = 0$}\end{sideways}}&
				\multicolumn{3}{c}{\includegraphics[width=0.95\textwidth]{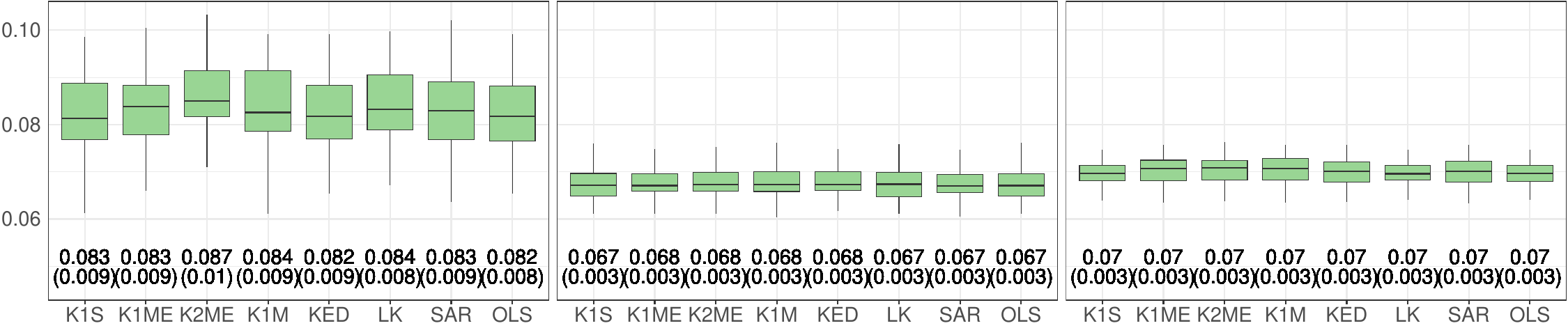}} \\
				{\begin{sideways}{\hspace{1.5cm}$\rho = 0.6$}\end{sideways}}&
				\multicolumn{3}{c}{\includegraphics[width=0.95\textwidth]{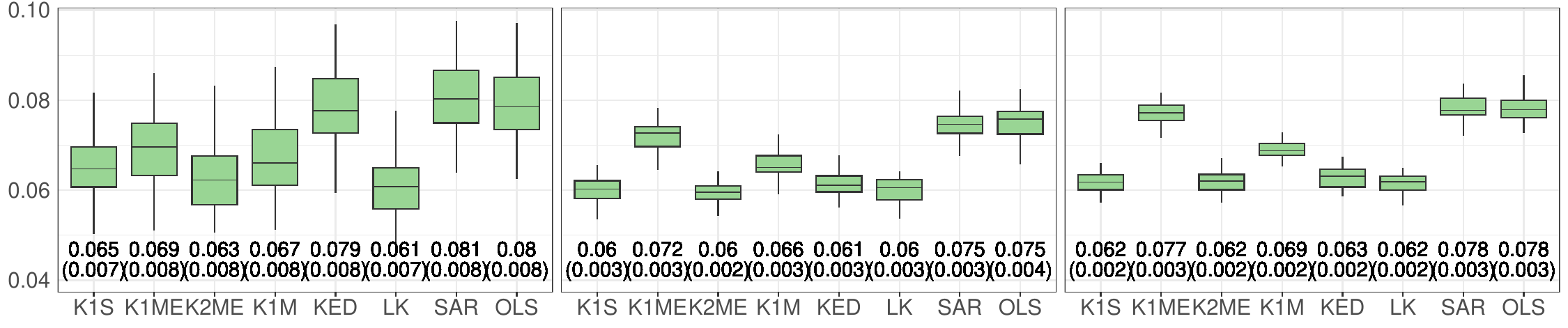}}\\		
				{\begin{sideways}{\hspace{1.5cm} $\rho = 0.9$}\end{sideways}}&
				\multicolumn{3}{c}{\includegraphics[width=0.95\textwidth]{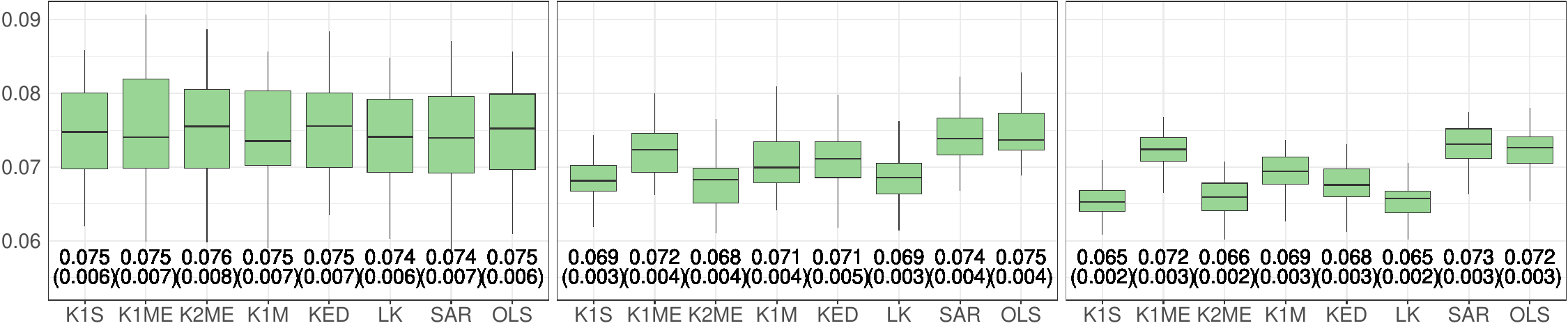}}			
			\end{tabular}
		\caption{RMSE distribution for the irregular grid.}
		\label{fig:regular_nXrho}
	\end{sidewaysfigure}

\subsubsection{Regular grid}
For the regular grid case (Figure \ref{fig:regular_nXrho}), each panel displays the nine combinations arising from three values of the autoregressive parameter, $\rho \in \{0, 0.6, 0.9\}$, and three sample sizes, $n \in \{100, 625, 1089\}$.
Overall, RMSE decreases as the grid size increases, with a more pronounced reduction in dispersion than in central tendency. This improvement is particularly evident when moving from $n = 100$ to larger sample sizes, whereas differences between $n = 625$ and $n = 1089$ are negligible and, in some cases, a slight increase in RMSE is observed.

The impact of the autoregressive parameter $\rho$ is more noticeable for small sample sizes. For $n = 100$, increasing $\rho$ from $0$ to $0.6$ or $0.9$ generally improves predictive performance across most methods, especially for the semiparametric approaches K1M, K2ME, and LK. In contrast, for $n = 625$ and $n = 1089$, the effect of $\rho$ becomes negligible, with similar RMSE values across methods. Notably, when $\rho = 0.6$, greater variability in performance is observed. This may reflect a scenario in which both covariates and spatial dependence contribute meaningfully to prediction, unlike the extreme cases where either covariates ($\rho = 0$) or the autoregressive component ($\rho = 0.9$) dominate. The intermediate case thus highlights differences in how effectively methods integrate these two sources of information.

In general, the proposed methodologies perform well. Excluding the case $\rho = 0$, K2ME exhibits the best predictive performance, with results comparable to LK. Additionally, K1S shows competitive performance across several scenarios under a regular grid.

\subsubsection{Irregular grid}
	\begin{sidewaysfigure}
			\centering
			\begin{tabular}{>{\columncolor[gray]{0.9}} c c  c c}
				\rowcolor[gray]{0.9}&  \multicolumn{3}{c}{\textbf{$n = 100 \hspace{5cm} n = 625 \hspace{5cm} n = 1089$}}\\
				{\begin{sideways}{\hspace{1.5cm}$\rho = 0$}\end{sideways}}&
				\multicolumn{3}{c}{\includegraphics[width=0.95\textwidth]{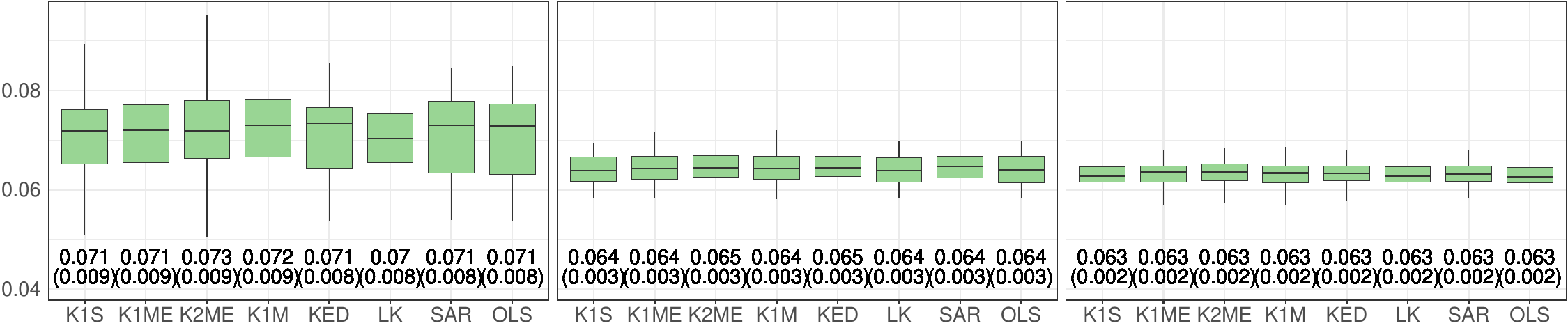}} \\
				{\begin{sideways}{\hspace{1.5cm}$\rho = 0.6$}\end{sideways}}&
				\multicolumn{3}{c}{\includegraphics[width=0.95\textwidth]{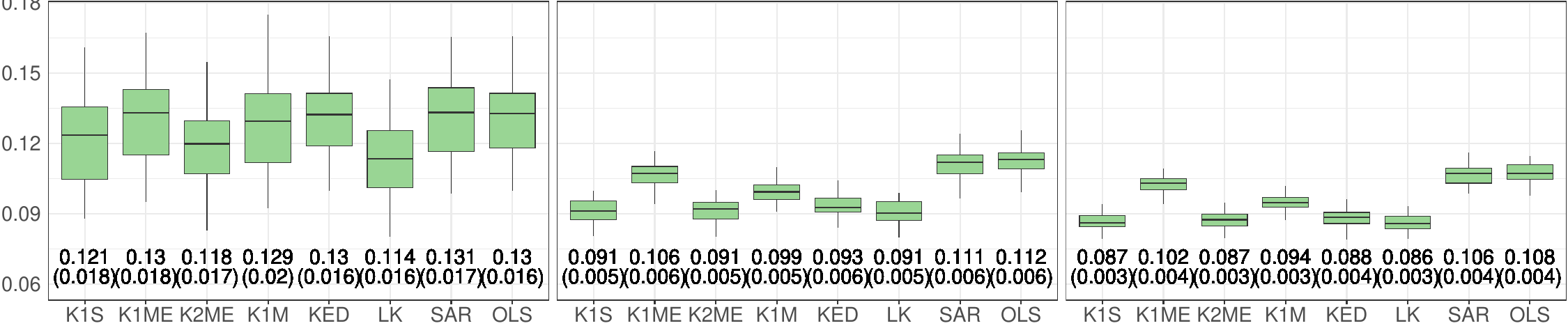}}\\
				{\begin{sideways}{\hspace{1.5cm} $\rho = 0.9$}\end{sideways}}&
				\multicolumn{3}{c}{\includegraphics[width=0.95\textwidth]{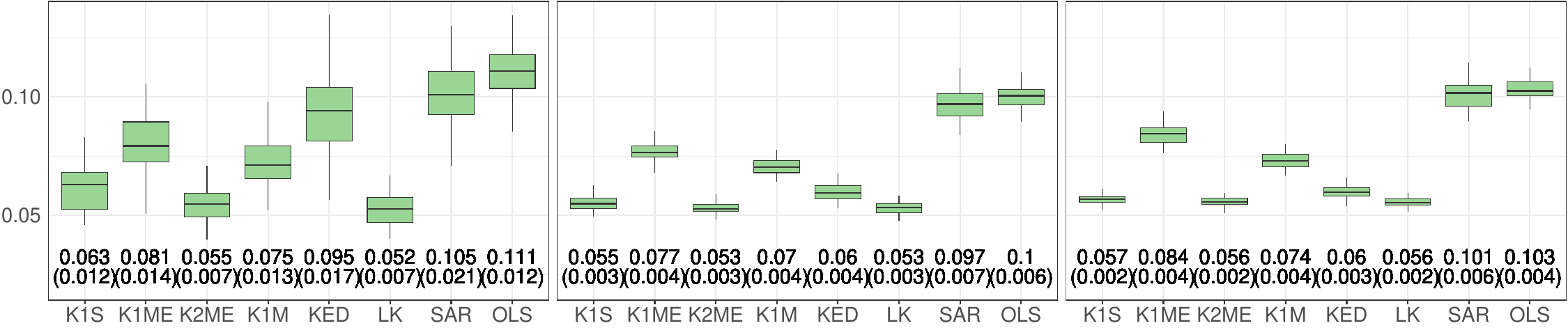}}
			\end{tabular}
		\caption{RMSE distribution for the irregular grid.}
		\label{fig:Irregular_nXrho}
	\end{sidewaysfigure}
Figure \ref{fig:Irregular_nXrho} presents the distribution of prediction errors under an irregular grid. Across all nine scenarios, semiparametric methods exhibit a similar relative ranking. The best performance is achieved by methods incorporating geographic distance in their weighting schemes. As the sample size increases, the availability of nearby observations improves, and kernels based on geographic distance—dominant in K1S and highly influential in K2ME—enhance predictive accuracy. In particular, K2ME and LK consistently yield the lowest RMSE in scenarios with spatial autocorrelation.

These results highlight the increasing relevance of spatial dependence, rendering OLS insufficient, as reflected by its consistently higher errors relative to semiparametric methods. Similarly, SAR shows poor predictive performance in most cases with $\rho \neq 0$, despite explicitly modeling autoregression.

Compared to the regular grid case (Figure \ref{fig:regular_nXrho}), greater variability across methods is observed, especially when $\rho = 0.9$. RMSE values are generally higher, particularly for $\rho = 0.6$, suggesting that irregular grid designs introduce additional estimation complexity, likely due to a reduced ability to capture local spatial correlation. For $\rho = 0.9$, some methods such as K1ME, SAR, and OLS also exhibit increased errors.

When $\rho = 0$, differences between grid designs become negligible, as expected, since spatial structure does not contribute to prediction. In contrast, for $\rho > 0$ and larger sample sizes, semiparametric methods outperform alternatives, achieving both lower RMSE and reduced variability. While kriging-based (parametric) methods improve as $\rho$ increases, they remain more sensitive to grid design. OLS consistently underperforms in the presence of spatial dependence, and SAR provides only marginal improvement, with predictive accuracy remaining close to that of OLS.

\subsubsection{Clustered grid}
    	\begin{sidewaysfigure}
			\begin{tabular}{>{\columncolor[gray]{0.9}}c p{0.3\textwidth} 					p{0.3\textwidth} p{0.3\textwidth}}
				\rowcolor[gray]{0.9}& \multicolumn{3}{c}{\textbf{$n = 100 \hspace{5cm} n = 625 \hspace{5cm} n = 1089$}} \\
				{\begin{sideways}\hspace{1.5cm}{$\rho = 0$}\end{sideways}}&
				\multicolumn{3}{c}{\includegraphics[width=0.95\textwidth]{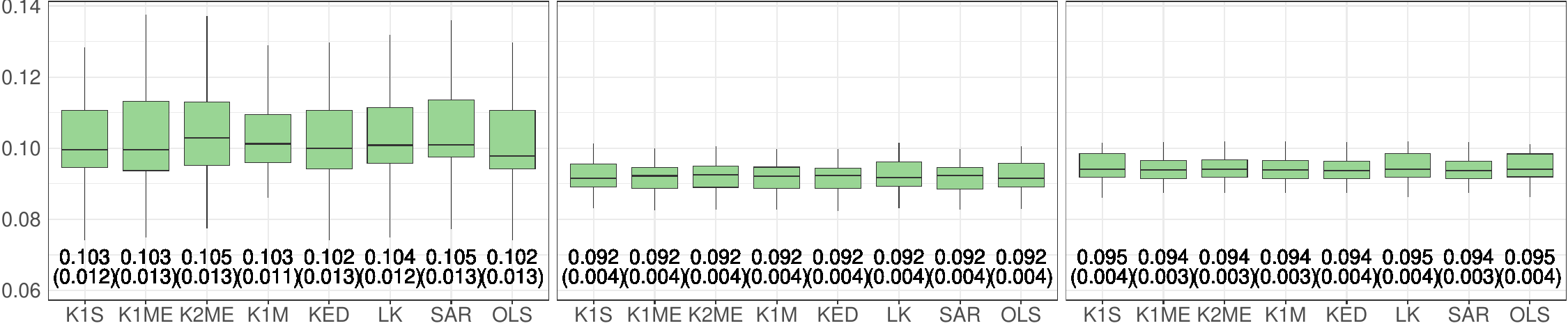}} \\
				{\begin{sideways}\hspace{1.5cm}{$\rho = 0.6$}\end{sideways}}&
				\multicolumn{3}{c}{\includegraphics[width=0.95\textwidth]{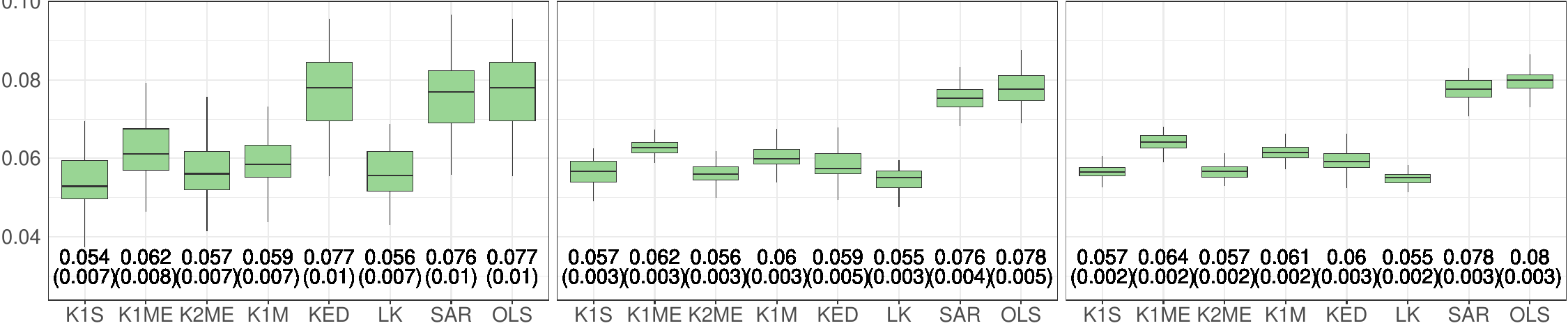}}\\			
				{\begin{sideways}\hspace{1.5cm}{ $\rho = 0.9$}\end{sideways}}&
				\multicolumn{3}{c}{\includegraphics[width=0.95\textwidth]{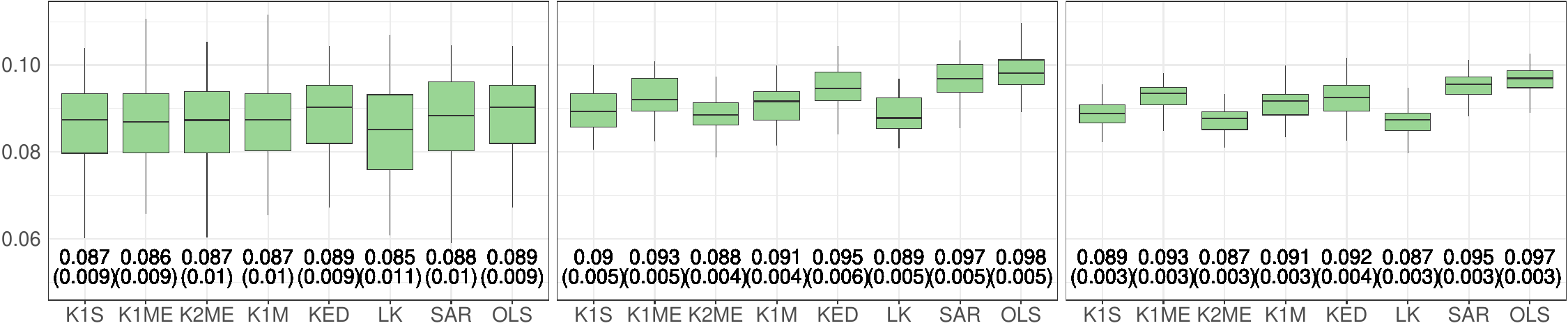}}			
			\end{tabular}
\caption{RMSE distribution for the clustered grid.}
\label{fig:Focos_nXrho}
	\end{sidewaysfigure}
Figure \ref{fig:Focos_nXrho} shows the same nine scenarios for the clustered grid design. As in the regular grid case, greater variability is observed when $\rho = 0.6$. In this setting, semiparametric methods together with LK achieve the best performance, with differences relative to other methods diminishing as the sample size increases.

Comparing the three grid designs—regular (Figure \ref{fig:regular_nXrho}), irregular (Figure \ref{fig:Irregular_nXrho}), and clustered (Figure \ref{fig:Focos_nXrho})—reveals how grid structure, spatial dependence, and methodology jointly affect predictive performance. RMSE decreases with increasing sample size ($n$) across all designs and values of $\rho$, reflecting improved estimation with more information.

Grid design plays a more substantial role when spatial dependence is present ($\rho > 0$). The irregular grid yields higher and more dispersed errors, particularly for small samples ($n = 100$), whereas the clustered design tends to produce lower RMSE, especially for moderate to high dependence ($\rho \ge 0.6$). This suggests that concentrated and highly correlated observations improving the ability of the methods to capture spatial dependence.

Across all grid designs, semiparametric methods consistently outperform alternatives, showing lower RMSE, reduced dispersion, and greater stability with respect to $n$ and $\rho$. Their advantage is most pronounced in settings with spatial dependence and larger samples, indicating a more effective integration of covariate and spatial information. Consequently, these methods provide a flexible and robust alternative when grid design and spatial correlation structure are not fully controlled.

\subsection{Results on estimating  $\texorpdfstring{\boldsymbol{\beta}}{\beta}$}
This section evaluates the predictive methods in terms of their ability to estimate the linear parameter of the proposed model. The comparison is restricted to approaches that explicitly include such estimation, namely the proposed semiparametric methods, lattice kriging (LK), the SAR model, and OLS.

Figure \ref{fig:betas} displays the mean differences between estimated and true coefficients across scenarios as the sample size increases, considering the three grid designs (regular, irregular, and clustered) and three levels of spatial autocorrelation ($\rho \in \{0, 0.6, 0.9\}$). For clarity, only K1S, K2ME, LK, OLS, and SAR are shown, while Table \ref{tablarotada} complements the analysis by including K1M and K1ME, reporting mean and standard deviation of $MAE_\beta$.

Overall, estimation error decreases with sample size, consistent with improved precision as more information becomes available, although the magnitude of this reduction varies across methods. Semiparametric approaches (K1S, K2ME) exhibit the most stable performance, with systematic error reduction and generally lower levels compared to alternatives. In contrast, OLS performs poorly, maintaining relatively large errors even for large samples, particularly when spatial dependence is present. SAR shows limited improvement over OLS.
\begin{figure}
	\centering
	\includegraphics[width=0.9\linewidth]{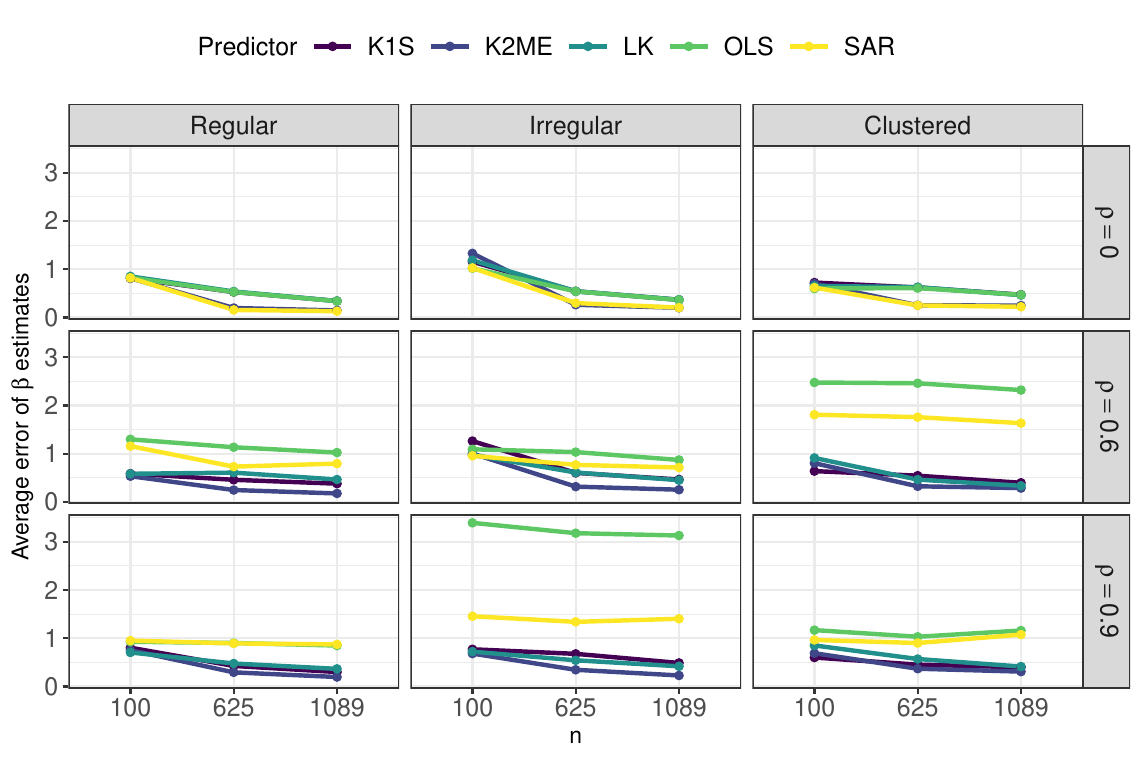}
	\caption{Average $MAE_{\beta}$ for different predictive methods.}
	\label{fig:betas}
\end{figure}
K2ME consistently achieves the lowest errors, with a clear decreasing trend and stable behavior across scenarios. K1S shows competitive performance but tends to be less stable, particularly in some grid configurations. LK improves with sample size, although its performance is less favorable in low-dependence settings ($\rho = 0$).

Differences across grid designs are present but secondary. The irregular grid generally leads to higher errors, especially for intermediate dependence ($\rho = 0.6$), whereas the regular grid exhibits smoother and more stable error reduction. 

When $\rho = 0$, all methods perform similarly, as spatial structure does not contribute to estimation. As $\rho$ increases, differences become more pronounced: semiparametric methods, and to a lesser extent LK, better exploit spatial dependence, while OLS and SAR improve more slowly and remain consistently less accurate.

Table \ref{tablarotada} further shows that K1S tends to produce higher errors in the absence of spatial dependence, particularly under regular and clustered grids, whereas K1M, K1ME, and K2ME behave similarly once moderate sample sizes are reached. For $\rho = 0.6$, performance becomes more heterogeneous, with some non-monotonic patterns in error reduction; K2ME remains the most stable, while K1S continues to lag. Under strong dependence ($\rho = 0.9$), overall errors increase and convergence slows, with K1ME showing instability in some grid designs, whereas K2ME maintains consistent improvements and the lowest error levels.

In summary, although no single method dominates uniformly, differences across methods become more relevant as spatial dependence increases. K2ME stands out for its robustness to changes in $\rho$ and grid design, while K1S and, in some cases, K1ME exhibit greater variability.
	\begin{sidewaystable}
		\centering
				\caption{$MAE_{\beta}$ over the 50 runs for each combination of $\rho$, $n$, and grid type.}
				\label{tablarotada}
			\begingroup
			\renewcommand{\baselinestretch}{1}
			\footnotesize
			\selectfont
			\renewcommand{\arraystretch}{1.5}
				\begin{tabular}[t]{lllllllllll}
					\toprule
					\multirow{2}{*}[-4mm]{$\rho$} &\multirow{2}{*}[-4mm]{Predictor} & \multicolumn{3}{c}{Regular}
					& \multicolumn{3}{c}{Irregular}
					& \multicolumn{3}{c}{Clustered}\\
					\cmidrule{3-11}
					&&$n=100$ & $n=625$ & $n=1089$ 
					&$n=100$ & $n=625$ & $n=1089$ 
					& $n=100$ & $n=625$& $n=1089$\\
					\midrule
					\multirow{6}{*}{$0$} 
					& K1M & 0.83 (0.19) & 0.16 (0.05) & \textbf{0.13} (0.04) & 1.11 (0.47) & 0.25 (0.1) & 0.18 (0.06) & 0.61 (0.24) & 0.23 (0.08) & 0.24 (0.08)\\
					& K1ME & 0.82 (0.19) & 0.17 (0.06) & 0.14 (0.05) & \textbf{1.02} (0.41) & 0.27 (0.11) & 0.19 (0.07) & \textbf{0.6} (0.25) & \textbf{0.22} (0.07) & 0.24 (0.08)\\
					& K1S & \textbf{0.81} (0.18) & 0.53 (0.24) & 0.34 (0.14) & 1.16 (0.52) & 0.54 (0.21) & 0.37 (0.13) & 0.72 (0.26) & 0.62 (0.27) & 0.47 (0.17)\\
					& K2ME & \textbf{0.81} (0.2) & 0.2 (0.05) & {0.14} (0.04) & 1.33 (0.48) & \textbf{0.26} (0.08) & \textbf{0.2} (0.07) & 0.7 (0.24) & 0.25 (0.08) & 0.24 (0.08)\\
					& LK & 0.85 (0.18) & 0.54 (0.23) & 0.34 (0.14) & 1.18 (0.39) & 0.54 (0.2) & 0.37 (0.13) & 0.65 (0.26) & 0.63 (0.27) & 0.46 (0.17)\\
					& OLS & 0.82 (0.18) & 0.53 (0.24) & 0.34 (0.14) & \textbf{1.02} (0.41) & 0.54 (0.21) & 0.36 (0.12) & \textbf{0.6} (0.24) & 0.61 (0.27) & 0.47 (0.17)\\
					& SAR & 0.82 (0.19) & \textbf{0.15} (0.03) & 0.13 (0.03) & 1.03 (0.42) & 0.3 (0.1) & \textbf{0.2} (0.06) & 0.62 (0.25) & 0.25 (0.07) & \textbf{0.22} (0.08)\\
					\addlinespace
					\multirow{6}{*}{$0.6$} 
					& K1M & 0.6 (0.15) & 0.41 (0.06) & 0.47 (0.04) & 0.92 (0.31) & 0.52 (0.14) & 0.44 (0.09) & 0.77 (0.2) & 0.55 (0.08) & 0.56 (0.07)\\
					& K1ME & 0.63 (0.2) & 0.64 (0.07) & 0.77 (0.05) & \textbf{0.73} (0.23) & 0.64 (0.16) & 0.6 (0.11) & 0.91 (0.28) & 0.73 (0.09) & 0.83 (0.08)\\
					& K1S & 0.59 (0.17) & 0.46 (0.17) & 0.38 (0.15) & 1.26 (0.37) & 0.6 (0.2) & 0.47 (0.17) & \textbf{0.64} (0.18) & 0.54 (0.13) & 0.4 (0.1)\\
					& K2ME & \textbf{0.53} (0.13) & \textbf{0.25} (0.08) & \textbf{0.18} (0.05) & 1 (0.33) & \textbf{0.32} (0.07) & \textbf{0.25} (0.06) & 0.8 (0.25) & \textbf{0.32} (0.08) & \textbf{0.29} (0.06)\\
					& LK & 0.58 (0.1) & 0.61 (0.24) & 0.46 (0.16) & 0.98 (0.34) & 0.61 (0.17) & 0.45 (0.15) & 0.91 (0.14) & 0.46 (0.14) & 0.34 (0.1)\\
					& OLS & 1.3 (0.12) & 1.13 (0.24) & 1.02 (0.19) & 1.09 (0.26) & 1.03 (0.18) & 0.87 (0.15) & 2.48 (0.15) & 2.46 (0.19) & 2.32 (0.12)\\
					& SAR & 1.16 (0.23) & 0.73 (0.05) & 0.79 (0.04) & 0.96 (0.28) & 0.77 (0.14) & 0.71 (0.1) & 1.81 (0.29) & 1.76 (0.09) & 1.63 (0.06)\\
					\addlinespace
					\multirow{6}{*}{$0.9$}  & K1M & 0.87 (0.18) & 0.73 (0.08) & 0.67 (0.06) & 1.21 (0.3) & 0.9 (0.13) & 0.93 (0.1) & 0.68 (0.25) & 0.47 (0.05) & 0.53 (0.06)\\
					& K1ME & 0.93 (0.19) & 0.86 (0.08) & 0.84 (0.05) & 1.41 (0.33) & 1.05 (0.12) & 1.31 (0.06) & 0.72 (0.19) & 0.69 (0.09) & 0.86 (0.08)\\
					& K1S & 0.8 (0.16) & 0.43 (0.14) & 0.3 (0.09) & 0.77 (0.19) & 0.68 (0.19) & 0.49 (0.17) & \textbf{0.6} (0.22) & 0.45 (0.18) & 0.4 (0.13)\\
					& K2ME & 0.75 (0.2) & \textbf{0.29} (0.06) & \textbf{0.2} (0.05) &\textbf{ 0.68} (0.13) &\textbf{ 0.35} (0.08) & \textbf{0.23} (0.06) & 0.69 (0.23) & \textbf{0.37} (0.1) & \textbf{0.31} (0.07)\\
					& LK & \textbf{0.7} (0.15) & 0.48 (0.12) & 0.36 (0.1) & 0.72 (0.11) & 0.54 (0.19) & 0.42 (0.15) & 0.85 (0.15) & 0.57 (0.2) & 0.41 (0.17)\\
					& OLS & 0.93 (0.14) & 0.9 (0.1) & 0.85 (0.07) & 3.39 (0.2) & 3.18 (0.27) & 3.13 (0.23) & 1.17 (0.15) & 1.03 (0.11) & 1.16 (0.1)\\
					& SAR & 0.95 (0.18) & 0.89 (0.07) & 0.87 (0.05) & 1.46 (0.27) & 1.34 (0.08) & 1.4 (0.05) & 0.96 (0.17) & 0.9 (0.09) & 1.08 (0.06)\\
					\bottomrule
				\end{tabular}
			\endgroup
	\end{sidewaystable}

\section{Application to real data}
Here we presents an application of the proposed methods to real data, comparing their performance both among themselves and against classical parametric predictors and kriging-based approaches. In this section, we follow the same methodology defined in Section 4.1.

The empirical analysis focuses on the prediction of school performance scores using data from 1965 primary schools in Ohio for the academic year 2001--2002 \citep{LeSage2009}. The objective is to predict an index of average student performance (\texttt{pscore}) at the district level.

Explanatory variables include educational system indicators and socioeconomic characteristics. The former comprise enrollment (\texttt{enroll}), number of teachers (\texttt{teachers}), average teaching experience (\texttt{experience}), average teacher salary (\texttt{salary}), student--teacher ratio (\texttt{ppupil}), and per-pupil expenditure (\texttt{PPS}), disaggregated into instruction (\texttt{pinstruct}), building operations (\texttt{pbuilding}), administration (\texttt{padminist}), student support (\texttt{ppsupport}), and staff support (\texttt{pssupport}). Additionally, expenditure shares are considered for instruction (\texttt{instructp}), building operations (\texttt{buildingp}), administration (\texttt{administp}), student support (\texttt{psupportp}), and staff support (\texttt{ssupportp}). Socioeconomic variables include per capita income (\texttt{logpincome}), proportion of nonwhite population (\texttt{nonwhite}), poverty rate (\texttt{poverty}), residential stability (\texttt{samehouse}), and proportion of students attending public schools (\texttt{public}). Educational attainment is captured through the shares of the population over 25 with high school (\texttt{highsh}), associate (\texttt{assoc}), college (\texttt{college}), graduate (\texttt{grad}), and professional (\texttt{prof}) education.

Spatial coordinates correspond to ZIP codes (latitude and longitude), with multiple schools sometimes sharing the same location. In such cases, variables are averaged to obtain a single observation per ZIP code, resulting in a final sample of 799 spatial units. 

\subsection{Results}
\begin{figure}
\includegraphics[width=.99\textwidth]{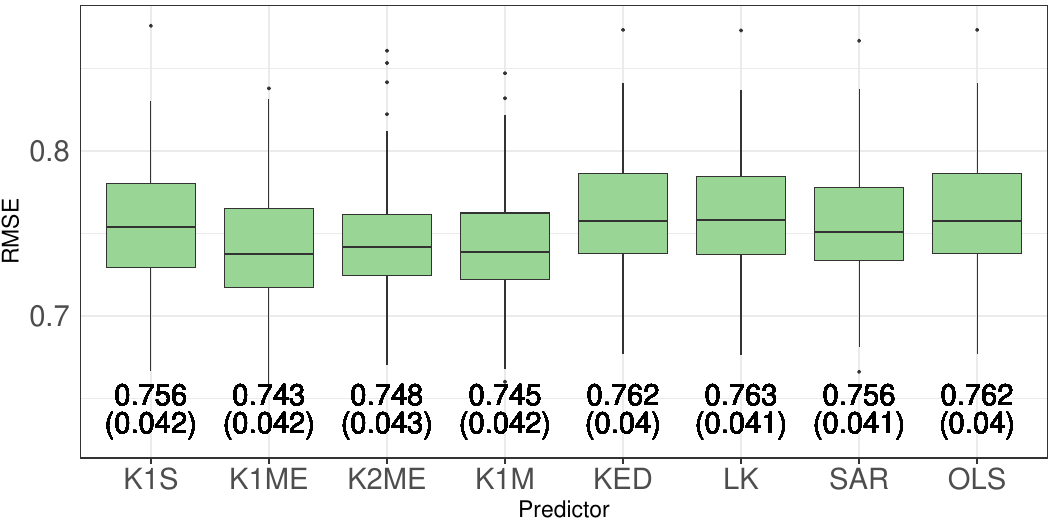}	
	\caption{RMSE distribution for the Ohio data.} \label{fig:RMSE_ohio}
\end{figure}
Figure \ref{fig:RMSE_ohio} shows that RMSE values are concentrated around 0.74--0.76. As we can observe,  the proposed approaches achieve slightly better predictive performance, with errors even lower than those of LK, which had shown highly competitive results in the simulation study. This proximity suggests that, for the Ohio dataset, both flexible and simpler models capture a substantial portion of the response variability. This is consistent with the exploratory analysis, which indicates clear associations between the response and covariates, together with a spatial structure that is present but not strongly dominant.

KED exhibits a slightly higher median and greater dispersion, indicating increased sensitivity to extreme observations or specific sampling configurations. However, these differences are minor and do not affect the overall ranking. LK performs similarly to kernel-based approaches, reinforcing the idea that local smoothing is sufficient to capture the relevant variability.

In contrast to previous settings, SAR and OLS do not show a marked deterioration in performance. This suggests that global spatial dependence and linear relationships provide a reasonable approximation to the underlying process in this case. Overall, the results indicate that method choice has a limited impact on RMSE, with only moderate gains from increased model complexity, in line with the exploratory findings.

\section{Conclusion and Discussion}

This work develops a semiparametric spatial autoregressive model for prediction, combining an interpretable linear component with a flexible nonparametric term while preserving spatial dependence. 

From a theoretical perspective, asymptotic properties of the estimators for both the linear parameters and the nonparametric component are established under standard regularity conditions. In particular, asymptotic normality of the linear parameter is derived, yielding convergence rates for both components and ensuring the consistency of the resulting predictor. These results provide formal support for the proposed estimation procedure and its large-sample behavior.

Simulation studies complement the theoretical analysis and show that the semiparametric approach achieves competitive predictive performance relative to classical parametric methods, especially when the functional form is complex or partially misspecified. At the same time, the results highlight that predictive accuracy is sensitive to the degree of spatial dependence and to the sampling design, underscoring the importance of careful exploratory analysis prior to model implementation.

The empirical application demonstrates that the proposed methodology is computationally feasible and adaptable to real-world data. The semiparametric model captures patterns not always identified by purely parametric approaches, without compromising predictive accuracy. 

Overall, the proposed semiparametric spatial autoregressive model and its estimation algorithm constitute a useful addition to the set of tools for spatial prediction. Its main advantage lies in its flexibility under functional uncertainty, while avoiding the need to specify covariance structures or spatial weight matrices.

Future research directions include extending the framework to nonstationary settings and more general dependence structures, as well as relaxing assumptions on the sampling design by allowing random spatial locations. Another relevant extension is the development of variance estimators for the linear parameters, enabling formal inference, either through consistent analytical estimators or bootstrap procedures adapted to dependent data.

\par
\section*{Acknowledgements}

This work was partially supported by the ANPCYT grant PICT-2019-00301 and the UNL grant 8552024010055LI.
\par


\bibhang=1.7pc
\bibsep=2pt
\fontsize{9}{14pt plus.8pt minus .6pt}\selectfont
\renewcommand\bibname{\large \bf References}
\expandafter\ifx\csname
natexlab\endcsname\relax\def\natexlab#1{#1}\fi
\expandafter\ifx\csname url\endcsname\relax
\def\url#1{\texttt{#1}}\fi
\expandafter\ifx\csname urlprefix\endcsname\relax\def\urlprefix{URL}\fi



\appendix
\section{Proofs of Theorems 1 and 2}\label{pruebas}
\subsection{Some results for the proofs}

\begin{lemma}\label{lemamixing}
	Let \(\{Z(\s) : \s \in \mathbb{Z}^d\}\) be a random field with spatial mixing coefficient defined, for $r\ge0$, by
	\[\label{alfa}
	\alpha(r)
	\defi
	\sup\Big\{
	\big|\mathbb P(A\cap B)-\mathbb P(A)\mathbb P(B)\big| :
	A\in\sigma(Z(s):s\in S),\;
	B\in\sigma(Z(s):s\in S'),\;
	\operatorname{dist}(S,S')\ge r
	\Big\},
	\]
	where $S,S'\subset\mathbb Z^d$ are finite sets, and
	\(
	\operatorname{dist}(S,S')
	\defi\inf\{\|s-s'\|: s\in S,\ s'\in S'\}.
	\)
	Assume that:
	\begin{enumerate}
		\item \label{1} There exist constants \(C > 0\) and \(b > 1\) such that, for all $r > 0$,
		\(
		\alpha(r) \leq C \, r^{-b}.
		\)
		\item \label{2} The domain $\mathcal D_\nn  \subset \mathbb Z^d$, for constants \(0 < c_L \leq C_L\), satisfies
		\(
		c_L \hat{n}^{1/d} \leq \mathcal D_\nn \leq C_L \hat{n}^{1/d},
		\)
		or equivalently, \(\mathcal D_\nn \asymp \hat{n}^{1/d}\).

		\item \label{3} The sets $E, E' \subset \mathbb{Z}^d$ have bounded density, i.e., there exists a constant $c_d>0$ such that, for any ball $B(\bm{x},r)\subset \mathbb{Z}^d$ with $r\ge 1$,
		\(\#\bigl(E \cap B(\bm{x},r)\bigr)\le c_d\, r^d
		\quad \text{and} \quad
		\#\bigl(E' \cap B(\bm{x},r)\bigr)\le c_d\, r^d.
		\)
	\end{enumerate}

	Then, for any fixed \(\varepsilon > 0\), there exists a constant \(C'\) such that, for all pairs of subsets \(E, E' \subset \mathcal{D}_\nn\) with \(\mathrm{dist}(E,E') \geq \varepsilon \hat{n}^{1/d}\), it holds that
	\[
	\alpha\bigl(\sigma(E), \sigma(E')\bigr) \leq C' \, \hat{n}^{-b/d}.
	\]
\end{lemma}

\begin{proof}
	Let $E, E'$ be subsets of $\mathcal D_\nn$ and fix $\varepsilon>0$. 
	Condition \ref{3} excludes local concentrations of points, while the expansion of the domain $\mathcal D_\nn$ of order $\hat{n}^{1/d}$ (condition \ref{2}) guarantees the existence of subsets separated by a distance of the same order.  
	Then, applying directly hypothesis \ref{1} with \(r = \mathrm{dist}(E,E') \geq \varepsilon \hat{n}^{1/d}\), we have
	\[
	\alpha\bigl(\sigma(E),\sigma(E')\bigr) 
	\leq C \bigl(\mathrm{dist}(E,E')\bigr)^{-b}
	\leq C \bigl(\varepsilon \hat{n}^{1/d}\bigr)^{-b}
	= C'\, \hat{n}^{-b/d},
	\]
	which proves the lemma.
\end{proof}
\begin{remark}
	The lemma states that, under a polynomial decay of the mixing coefficient and regular growth of the domain, two subsets of sites separated by a distance of order $\hat n^{1/d}$ within the domain exhibit weak dependence, with no local clustering effects.
\end{remark}
\begin{remark}
	Lemma \ref{lemamixing} links assumption \ref{mixing} with the mixing coefficient condition in Lemma 3 from \cite{Aneiros-Perez2008Nonparametric}, ensuring its applicability.
\end{remark}
\begin{lemma}\label{lemaultimo}
	Let $\{(Y_i,\Xbf_i^{T},\bm{T}_i^{T})\}$ be a sample satisfying assumptions A3, A4, A5, and A6.
	Suppose the data follow the model \eqref{modeloespacial}. If $\bm \eta_i = \mathbb{E}[\bm{X}_i - \mathbb{E}[\bm{X}_i \mid \T_i]]$, then $\{\varepsilon_i\}$, $\{\bm\eta_i\}$, and $\{\bm\eta_i\varepsilon_i\}$ have zero mean, are strictly stationary, and strongly mixing, with mixing coefficients bounded by $\alpha(r)$. Furthermore, there exist $r>4$, $\delta > 0$, and a constant $C>0$ such that
	\[
	\sup_{i} \mathbb{E}[|\varepsilon_i|^{r}] \le C < \infty,
	\]
	\[\label{momentReta}
	\sup_{i} \mathbb{E}[\|\bm\eta_i\|^{r}] \le C < \infty,
	\]
	and
	\[
	\mathbb{E}[\|\bm\eta_i\varepsilon_i\|^{2+\delta}] < \infty,
	\]
	respectively.
\end{lemma}
\begin{proof}
	First, we show that $\{\varepsilon_i\}$, $\{\bm\eta_i\}$, and $\{\bm\eta_i\varepsilon_i\}$ have zero mean. Since $\mathbb{E}[\varepsilon_i \mid \Xbf_i, \bm{T}_i] = 0$ by \eqref{modeloespacialhope}, it follows that
	\(
	\mathbb{E}[\varepsilon_i] = \mathbb{E}[\mathbb{E}[\varepsilon_i \mid \Xbf_i, \bm{T}_i]] = 0,
	\)
	and similarly, for each component $\eta_{ij}$,
	\(
	\mathbb{E}[\eta_{ij}] = \mathbb{E}[X_{ij} - \mathbb{E}[X_{ij}\mid \bm{T}_i]] = 0.
	\)
	Independence of $\bm\eta_i$ and $\varepsilon_i$ (A6) then gives $\mathbb{E}[\bm\eta_i \varepsilon_i] = \bm{0}$.

	Stationarity of $\{\bm\eta_i\}$ follows since both $\bm{X}_i$ and $\mathbb{E}[\bm{X}_i\mid \T_i]$ are measurable functions of $\{(\bm{X}_i,\T_i)\}$. Similarly, $\{\varepsilon_i\}$ is stationary as a measurable function of $\{(Y_i, \bm{X}_i^T, \bm{T}_i^T)\}$, and so is $\{\bm\eta_i \varepsilon_i\}$.

	By A5, the joint process $\{(Y_i,\bm{X}_i^T,\bm{T}_i^T)\}$ is strongly mixing with coefficients $\alpha(r) \le C r^{-b}$. Since $\varepsilon_i$ and $\eta_{ij}$ are measurable functions of this process, their mixing coefficients are bounded by $\alpha(r)$.

	Regarding moments, A\ref{fs2}(ii) ensures that, for some $r>4$,
	\(
	\sup_i \mathbb{E}[|Y_i|^r] < \infty, \sup_i \mathbb{E}[|X_{ij}|^r] < \infty.
	\)
	Since $r(\cdot)$ and $g(\cdot)$ are Lipschitz on the compact set $\mathcal C$ (A4), they are uniformly bounded. Consequently, 
	\[
	\sup_i \mathbb{E}[|\varepsilon_i|^r] < \infty, \qquad 
	\sup_i \mathbb{E}[|\eta_{ij}|^r] < \infty.
	\]
	Finally, to show $\mathbb{E}[\|\bm\eta_i \varepsilon_i\|^{2+\delta}] < \infty$, note that
	\[
	\|\bm\eta_i \varepsilon_i\|^{2+\delta} = |\varepsilon_i|^{2+\delta} \|\bm\eta_i\|^{2+\delta} \le C \sum_{k=1}^p |\varepsilon_i|^{2+\delta} |\eta_{ik}|^{2+\delta}.
	\]
	By Hölder's inequality and choosing $0 < \delta < r/2 - 2$, the moments on the right-hand side are uniformly bounded, which proves the claim.
\end{proof}
\begin{lemma}\label{lema4}
    Assume that A1--A5 hold. Then,
\[\label{cotaPeso}
\max_{1\le i,j\le \hat{n}} \left| \omega_{\mathbf{h}}(\T_i,\T_j) \right|
= O\!\left((\hat{n} h_{1\hat{n}}^{d} h_{2\hat{n}}^{k})^{-1}\right)
\quad \text{a.s.},
\]
where $\omega_{\mathbf{h}}(\T_i,\T_j)$ is defined in \eqref{pesos0}.
\end{lemma}
\begin{proof}
	Recall that the weights defined in \eqref{pesos0} can be written as
	\[\label{pesosFraccion}
	\omega_{\h}(\tB,\T_i) = \frac{\Delta_i(\tB)}{\hat{q}(\tB)},
	\]
	where
	\(
	\Delta_i(\tB) = \frac{1}{a_{\nn}h^k_{2\nn}}K_{1, h_{1\nn}} \left( d(\s, \s_i) \right) K_2\left( {d_m(\bm{t}, \bm{T}_i)}/{h_{2\nn}} \right),
	\)
	and
	\(
	\hat{q}(\tB)=\frac{1}{a_{\nn}h^k_{2\nn}} \sum_{j=1}^{\hat{n}} K_{1, h_{1\nn}} \left( d(\s, \s_j) \right) K_2\left( {d_m(\bm{t}, \bm{T}_j)}/{h_{2\nn}} \right),
	\)
whit $a_{\nn}\defi  \sum_{j=1}^{\hat{n}}K_{1, h_{1\nn}} \left( d(\s, \s_j) \right). $
	We bound separately the numerator and the denominator.

	\textit{Step 1 (numerator).} By A\ref{K}, the kernels are bounded, i.e., $m_i \le K_i(u) \le M_i$ and 
    	\begin{equation}
		\label{cotaan}
	m_1\sum_{\s_i \in \mathcal{D}_\nn} \mathbb{I}_{\{d(\s, \s_i) \le {n}h_{1\nn}\}}\le 
	a_\nn \le
	M_1\sum_{\s_i \in \mathcal{D}_\nn} \mathbb{I}_{\{d(\s, \s_i) \le {n}h_{1\nn}\}}.
\end{equation} 
    Hence,
	\(
	\Delta_i(\tB) \le {M_1 M_2}/{a_{\nn}h^k_{2\nn}}.
	\)
	Using \eqref{cotaan}, \eqref{expansionk_n}, $r_\nn = n h_{1\nn}$, $\beta < d$ y $\hat{n} = n^d$ and the domain growth assumption, we have
	\(
	a_{\nn} \asymp \hat{n}h^d_{1\nn},
	\)
	which yields
	\[
	\Delta_i(\tB) = O\!\left(\frac{1}{\hat{n}h^d_{1\nn}h^k_{2\nn}}\right)
	\quad \text{a.s.}
	\]

	\textit{Step 2 (denominator).} From the proof of Theorem 3.1 (from \cite{DaboNiang2016}), $\hat q(\tB)$ converges uniformly almost surely to $q(\tB)$. Since $\inf_{\tB \in \mathcal C} |q(\tB)| \ge \delta > 0$ (A\ref{fs}), it follows that, for $\hat{n}$ large enough,
	\(
	\inf_{\tB\in\mathcal C} |\hat q(\tB)| \ge {\delta}/{2}
	\quad \text{a.s.}
	\)
	and therefore $|\hat q(\tB)|^{-1} \le 2/\delta$.

	Combining both bounds, we obtain $|\omega_{\h}(\tB,\T_i)|
	\le C \, (\hat{n}h^d_{1\nn}h^k_{2\nn})^{-1}
	\quad \text{a.s.}, $ which completes the proof.
\end{proof}

\begin{lemma}\label{lema5}
Under assumptions H1--H6, suppose that the sample 
$\{(X_{i1},\ldots,X_{ip},\bm{T}_i)\}$ 
arises from a strictly stationary $\alpha$-mixing process whose mixing coefficient satisfies 
$\alpha(r)\le c r^{-b}$, with $b>1$. Moreover, assume that there exists $r>4$ such that, for $ j=1,\dots,p, $ 
\(
\max_{1\le i \le n}\mathbb{E}(|X_{ij}|^r) \le C < \infty.
\)
Additionally, if $h_{2\hat{n}} \to 0$ and 
\(
\frac{\log \hat{n}}{\hat{n} h^d_{1\hat{n}} h^k_{2\hat{n}}} \to 0
\quad \text{as } \hat{n} \to \infty,
\)
then
\[
\hat{n}^{-1}\tilde{\mathbb{X}}^{\top}\tilde{\mathbb{X}} \longrightarrow \mathbf{B}
\quad \text{almost surely.}
\]
\end{lemma}
\begin{proof}
Since $ \tilde{\mathbb{X}}  = {\mathbb{X}}-\bm{W}_{\h}{\mathbb{X}} $, the {$ kj $}-th entry is given by
\begin{equation}\label{graya}
\tilde{X}_{kj}
= X_{kj} - \dsum{l}{1}{\hat{n}}  \omega_{\h}(\bm{T}_k,\bm{T}_{l} )X_{lj} 
\doteq \eta_{kj} - \bar{g}_{j}(\bm{T}_k),
\end{equation}
where $\eta_{kj} =  X_{kj} -\hope{X_{kj}|\bm{T}_k}$ and
$
\bar{g}_{j}(\bm{T}_k)
= \hope{X_{kj}|\bm{T}_k} - \dsum{l}{1}{\hat{n}}  \omega_{\h}(\bm{T}_k,\bm{T}_{l} )X_{lj}.
$

Thus, the {$js$}-th entry of $\hat{n}^{-1}\tilde{\mathbb{X}}^T \tilde{\mathbb{X}}$ can be written as
\begin{align}
\bigg(\hat{n}^{-1}\tilde{\mathbb{X}}^T \tilde{\mathbb{X}}\bigg)_{js} 
&=\hat{n}^{-1}\sum_{k=1}^{\hat{n}}  \tilde{X}_{kj}\tilde{X}_{ks}\\
&=\hat{n}^{-1}\bigg(\sum_{k=1}^{\hat{n}} \eta_{kj}\eta_{ks}
- \sum_{k=1}^{\hat{n}} \bar{g}_{j}(\bm{T}_k)\,\eta_{ks} 
- \sum_{k=1}^{\hat{n}} \bar{g}_{s}(\bm{T}_k)\,\eta_{kj} + \sum_{k=1}^{\hat{n}}  \bar{g}_{j}(\bm{T}_k)\,\bar{g}_{s}(\bm{T}_k) \bigg).
\end{align}

We first analyze the leading term. Since $\{\bm\eta_k\}$ is strictly stationary, so is $\{\eta_{kj}\eta_{ks}\}$ for Lemma \ref{lemaultimo}. Moreover, the strong mixing property is preserved under measurable transformations, hence $\{\eta_{kj}\eta_{ks}\}$ is also strongly mixing. 
Next, we verify the moment condition. Since
\(
|\eta_{ij}\eta_{is}|\le \|\bm\eta_i\|^2,
\)
it follows that
\(
\hope{|\eta_{ij}\eta_{is}|^{r/2}}
\le
\hope{\|\bm\eta_i\|^{r}}.
\)
By \eqref{momentReta}, there exists $r>4$ such that
\(
\hope{\|\bm\eta_i\|^{r}} \le C < \infty,
\)
and therefore $\hope{|\eta_{ij}\eta_{is}|^{r/2}}<\infty$. This implies
\[
\hope{|\eta_{ij}\eta_{is}|\log^+|\eta_{ij}\eta_{is}|}<\infty,
\]
so that, by Theorem 2(b) from \cite{Devroye2010}, $hat{n}^{-1}\sum_{k=1}^{\hat{n}} \eta_{kj}\eta_{ks}
\;\longrightarrow\; \textbf{B}_{js}
\quad\text{a.s.} $ In particular, $
\hat{n}^{-1}\sum_{k=1}^{\hat{n}} \eta_{ks}^2 = O(1)
\quad\text{a.s.} $

We now control the remaining terms. As in the definition for $\bar{g}_{j}(T_k)$ the weight is given by \eqref{pesos0},  $\bm{T}_k$ are the covariates, and that each $X_{kj}$ represents the variable to be predicted and taking into acount assumptions A1--A5, Observation~4 in \cite{DaboNiang2016} yields and
\begin{equation}\label{30graya}
\max_{1\le j \le p} \max_{1\le k \le \hat{n}} 
\left| \overline{g}_{j}(\T_k) \right|
= O(h_{2\nn}) 
+ O\!\left(\frac{\log \hat{n}}{\hat{n} h^d_{1\nn}h^k_{2\nn}} \right)
\quad \text{a.s.}
\end{equation}
which converges to zero under the bandwidth conditions. Hence, $
\hat n^{-1}\sum_{k=1}^{\hat n}\bar g_j(\T_k)^2 \to 0.$

Applying Cauchy--Schwarz,
\[
\left|
\hat n^{-1}\sum_{k=1}^{\hat n} \bar g_j(\T_k)\eta_{ks}
\right|
\le
\left(
\hat n^{-1}\sum_{k=1}^{\hat n}\bar g_j(\T_k)^2
\right)^{1/2}
\left(
\hat n^{-1}\sum_{k=1}^{\hat n}\eta_{ks}^2
\right)^{1/2},
\]
which implies that $\hat n^{-1}\sum_{k=1}^{\hat n} \bar g_j(\T_k)\eta_{ks}
\longrightarrow 0.$ The same argument applies to the remaining terms, which are therefore negligible. This concludes the proof.
\end{proof}

\subsection{Proofs}
\begin{proof}[Proof of Theorem \ref{teo1}]
Recall that \(\hat \betaB_{\h} = (\tilde{\mathbb{X}}^T\tilde{\mathbb{X}})^{-1} \tilde{\mathbb{X}}^T \tilde\YB\) where \(\tilde\YB = \bm{Y}-\bm{W}_{\h} \bm{Y}\) and \(\tilde{\mathbb{X}}= {\mathbb{X}}-\bm{W}_{\h}{\mathbb{X}}\). Then,
\begin{align}
\sqrt{n}(\hat \betaB_{\h} -\betaB ) &= 	\sqrt{n}\big((\tilde{\mathbb{X}}^T\tilde{\mathbb{X}})^{-1} \tilde{\mathbb{X}}^T \tilde\YB-\betaB\big)
\label{teo1.1}
= \sqrt{n}(\tilde{\mathbb{X}}^T\tilde{\mathbb{X}})^{-1}\big(\tilde{\mathbb{X}}^T (\bm{Y}-\bm{W}_{\h} \bm{Y})-\tilde{\mathbb{X}}^T\tilde{\mathbb{X}}\betaB\big).
\end{align}
Moreover, defining $\bm{r} \defi (r_1,\dots, r_{\hat{n}})^T$ with $r_i = r(\T_i)$ and $\bm{\varepsilon} = (\varepsilon_1, \cdots, \varepsilon_{\hat{n}})$, model \eqref{modeloespacial} can be written in matrix form as \( \YB = \mathbb{X}\betaB + \bm{r} + \bm{\varepsilon} \).
Replacing the model in \eqref{teo1.1}, we obtain
\begin{align}
\tilde{\mathbb{X}}^T (\bm{Y}-\bm{W}_{\h} \bm{Y})-\tilde{\mathbb{X}}^T\tilde{\mathbb{X}}\betaB
\label{descomposicion}
&= \tilde{\mathbb{X}}^T{\mathbb{X}}\betaB- \tilde{\mathbb{X}}^T\bm{W}_{\h}\mathbb{X}\betaB   -\tilde{\mathbb{X}}^T\tilde{\mathbb{X}}\betaB\\
\label{rdesc}
& \qquad + \tilde{\mathbb{X}}^T\bm{r}- \tilde{\mathbb{X}}^T\bm{W}_{\h}\bm{r} \\
& \qquad + \tilde{\mathbb{X}}^T\bm{\varepsilon} - \tilde{\mathbb{X}}^T\bm{W}_{\h}\bm{\varepsilon}.
\end{align}
Then, since
\(
\tilde{\mathbb{X}}^T\mathbb{X}\betaB - \tilde{\mathbb{X}}^T\bm{W}_{\h}\mathbb{X}\betaB =
\tilde{\mathbb{X}}^T(\mathbb{X}-\bm{W}_{\h}\mathbb{X})\betaB
=  \tilde{\mathbb{X}}^T\tilde{\mathbb{X}}\betaB,
\)
it follows from \eqref{descomposicion} that
\(
\tilde{\mathbb{X}}^T\mathbb{X}\betaB - \tilde{\mathbb{X}}^T\bm{W}_{\h}\mathbb{X}\betaB -  \tilde{\mathbb{X}}^T\tilde{\mathbb{X}}\betaB=0.
\)
On the other hand, defining \(\tilde{\bm{r}}  \defi \bm{r}-\bm{W}_{\h} \bm{r}\), we can write \eqref{rdesc} as
\(
\tilde{\mathbb{X}}^T\bm{r}- \tilde{\mathbb{X}}^T\bm{W}_{\h}\bm{r} = \tilde{\mathbb{X}}^T\tilde{\bm{r}}.
\)
Using these two equalities, we obtain
\[
\tilde{\mathbb{X}}^T (\bm{Y}-\bm{W}_{\h} \bm{Y})-\tilde{\mathbb{X}}^T\tilde{\mathbb{X}}\betaB = \tilde{\mathbb{X}}^T\tilde{\bm{r}}+\tilde{\mathbb{X}}^T\bm{\varepsilon} - \tilde{\mathbb{X}}^T\bm{W}_{\h}\bm{\varepsilon}.
\]
Finally, in \eqref{teo1.1}, it follows that
\begin{align}
\sqrt{\hat{n}}(\hat \betaB_{\h} -\betaB ) &= 
\sqrt{\hat{n}}(\tilde{\mathbb{X}}^T\tilde{\mathbb{X}})^{-1}
\big(\tilde{\mathbb{X}}^T\tilde{\bm{r}}-
\tilde{\mathbb{X}}^T\bm{W}_{\h}\bm{\varepsilon}+
\tilde{\mathbb{X}}^T\bm{\varepsilon} \big)\\ \label{31}
&\defi(\hat{n}^{-1}\tilde{\mathbb{X}}^T\tilde{\mathbb{X}})^{-1}\hat{n}^{-\frac{1}{2}}(\bm{S}_{n1}-\bm{S}_{n2}+\bm{S}_{n3}).		
\end{align}
For the study of the asymptotic behavior of $\bm{S}_{n1}$, $\bm{S}_{n2}$ and $\bm{S}_{n3}$, vectors in $\rr^p$, we first consider the decomposition of the $ij$-th element of the matrix $ \tilde{\mathbb{X}} $, which appears in all terms. To this end, adding and subtracting $ \tilde{g}_{j}(\bm{T}_i) \defi \hope{X_{ij}|\bm{T}_i} -\dsum{k}{1}{\hat{n}}  \omega_{\h}(\bm{T}_i,\bm{T}_{k} )\hope{X_{kj}|\bm{T}_k} $ and using $\eta_{ij} = X_{ij} - \hope{X_{ij}|\T_i}$ defined in H\ref{epsilonEta}, we obtain
\begin{align}
\tilde{X}_{ij} &= X_{ij} -\dsum{k}{1}{\hat{n}} \omega_{\h}(\bm{T}_i,\bm{T}_{k} )X_{kj}\\
&= \hope{X_{ij}|\bm{T}_i} -\dsum{k}{1}{\hat{n}}  \omega_{\h}(\bm{T}_i,\bm{T}_{k} )\hope{X_{kj}|\bm{T}_k}
+	X_{ij}-\hope{X_{ij}|\bm{T}_i} - \dsum{k}{1}{\hat{n}}   \omega_{\h}(\bm{T}_i,\bm{T}_{k} )\eta_{kj}\\
\label{gPelito}
&= \tilde{g}_{j}(\bm{T}_i)+\eta_{ij}- \dsum{k}{1}{\hat{n}} \omega_{\h}(\bm{T}_i,\bm{T}_{k} )\eta_{kj}.
\end{align}
Hence, the $j$-th component of 
$
\bm{S}_{n1} = \tilde{\mathbb{X}}^T\tilde{\bm{r}}
= \dsumi\tilde{\bm{X}}_i\tilde{r}(\T_i)
$
admits the decomposition
\begin{align}
S_{n1,j} &= \dsumi \tilde{g}_{j}(\bm{T}_i)\tilde{r}(\T_i)+
\dsumi\eta_{ij}\tilde{r}(\T_i)-
\dsumi\tilde{r}(\T_i)\bigg(\dsum{k}{1}{\hat{n}} \omega_{\h}(\bm{T}_i,\bm{T}_{k} )\eta_{kj}\bigg)\defi S_{n1,j1} + S_{n1,j2} - S_{n1,j3}. 
\end{align}
Analogously, for the $j$-th component of 
$\bm{S}_{n2}$ and $\bm{S}_{n3}$.

In what follows, Lemma 3 from \cite{Aneiros-Perez2008Nonparametric} is applied to 
$\{\varepsilon_k\}$ and $\{\eta_{kj}\}$. To this end, we verify that both processes satisfy its assumptions.
From assumptions A3, A4, A5 and A6, together with Lemma \ref{lemaultimo}, it follows that the sequences $\{\varepsilon_i\}$ and $\{\bm\eta_i\}$ have zero mean, are strictly stationary, and are strongly mixing with mixing coefficients bounded by $\alpha(r)$. Moreover, Lemma \ref{lemaultimo} ensures the existence of $r>4$ and a constant $C>0$ such that
\(
\max_{1 \le i \le n} \hope{|\varepsilon_i|^{r}} \le C < \infty,
\)
and
\(
\max_{1 \le i \le n} \hope{\|\bm\eta_i\|^{r}} \le C < \infty.
\)
On the other hand, Lemma~\ref{lemamixing} ensures that
\[
\sum_{\hat{n}=1}^\infty \hat{n}^{\frac{5+4\gamma}{4(1-\gamma)}}\alpha(\hat{n})
\leq \sum_{\hat{n}=1}^\infty \hat{n}^{\frac{5+4\gamma}{4(1-\gamma)}} \hat{n}^{-b/d},
\]
and, by assumptions~A6, $b>4.5\,d$. Choosing $0.5<\gamma<1-\frac{9}{4b}$, the series is convergent. Therefore, both processes satisfy all the conditions required for the application of Lemma 3 from \cite{Aneiros-Perez2008Nonparametric}.
Moreover, Lemma \ref{lema4} implies that
\(
\max_{1\le i,j\le \hat{n}} |\omega_{\h}(\T_i,\T_j)| = O((\hat{n}h^d_{1\nn}h^k_{2\nn})^{-1}), \text{   a.s.}
\)
Thus, applying Lemma 3 from \cite{Aneiros-Perez2008Nonparametric} with
\(a_{ik} = \omega_{\h}(\T_i, \T_k)\),
\(a_{\nn} = (\hat{n}h^d_{1\nn}h^k_{2\nn})^{-1}\),
\(V_k = \varepsilon_k\), and
\(0.5 < \gamma < 1 - 9/(4b)\), we obtain
\[ \label{32}
\max_i \left| \sum_{k=1}^{\hat{n}} \omega_{\h}(\T_i, \T_k)\varepsilon_k \right|
= O\left(( h^d_{1\nn}h^k_{2\nn})^{-1} \hat{n}^{-1/2 + 1/r} \log \hat{n} \right)
\quad \text{a.s.}
\]
Under the same conditions, taking \(V_k = \eta_{kj}\), we obtain
\[\label{33eta}
\max_i \left| \sum_{k=1}^{\hat{n}} \omega_{\h}(\T_i, \T_k)\eta_{kj} \right|
= O\left( ( h^d_{1\nn}h^k_{2\nn})^{-1} \hat{n}^{-1/2 + 1/r} \log \hat{n} \right)
\quad \text{a.s.}
\]
both for $r>4$.

On the other hand, using the equality in \eqref{gPelito}, the definition of $\eta_{kj}$ and $\overline{g}_{j}(\T_i)$, we have
\begin{align}
\tilde{g}_{j}(\bm{T}_i) 
&= \overline{g}_{j}(\T_i) +\dsum{k}{1}{\hat{n}} \omega_{\h}(\bm{T}_i,\bm{T}_{k} )\eta_{kj},
\end{align}
where $\overline{g}(\T_i)$ is defined in \eqref{graya}.
Finally, combining \eqref{30graya} and \eqref{33eta}, we obtain
\begin{align}\label{35}
\max_{i,j} |\tilde g_{j}(\T_i)|
&= O(h_{2\nn})
+ O\left( (h^d_{1\nn}h^k_{2\nn})^{-1} \hat{n}^{-1/2 + 1/r} \log \hat{n} \right)
\quad \text{a.s.}
\end{align}
Similarly, we bound \(\max_{i} |\tilde{r}(\T_i)|\). Recall that \(\tilde{\bm{r}} =\bm{r}-\bm{W}_{\h} \bm{r}\), so the $i$-th element of \(\tilde{\bm{r}}\) is given by
\(
\tilde{r}(\T_i)= r(\T_i)-\sum_{k=1}^{\hat{n}} \omega_{\h}(\T_i, \T_k)r(\T_k).
\)
Moreover, from the model definition in \eqref{modeloespacial}, we have $r(\T_k) = Y_k - \bm{X}^T_{k}\betaB - \varepsilon_k$. On the other hand, from the conditional expectation in \eqref{poblacional}, $r(\T_k) = \hope{Y_i - \bm{X}^T_{i}\betaB|\T_i}$. Hence, the previous expression can be written as
\[
\tilde{r}(\T_i) = r(\T_i)-\sum_{k=1}^{\hat{n}} \omega_{\h}(\T_i, \T_k)r(\T_k)
= \hope{Y_i - \bm{X}^T_{i}\betaB|\T_i} -\sum_{k=1}^{\hat{n}} \omega_{\h}(\T_i, \T_k)\big(Y_k - \bm{X}^T_{k}\betaB - \varepsilon_k\big)
\]
Adding and subtracting $ \betaB^T\sum_{k=1}^{\hat{n}} \omega_{\h}(\T_i, \T_k)\hope{\bm{X}_k |\T_k}$, we obtain
\begin{align}
\label{m}
\tilde{r}(\T_i) 
&= \hope{Y_i |\T_i} -\sum_{k=1}^{\hat{n}} \omega_{\h}(\T_i, \T_k)Y_k \\
\label{eta}
&\quad + \betaB^T\sum_{k=1}^{\hat{n}} \omega_{\h}(\T_i, \T_k)\bigg(\bm{X}_{k}-\hope{\bm{X}_k |\T_k}\bigg)\\
\label{gPelito2}
&\quad - \betaB^T\bigg(\hope{\bm{X}_i |\T_i} - \sum_{k=1}^{\hat{n}} \omega_{\h}(\T_i, \T_k)\hope{\bm{X}_k |\T_k}\bigg)\\
\label{epsilon}
&\quad + \sum_{k=1}^{\hat{n}} \omega_{\h}(\T_i, \T_k)\varepsilon_k.		
\end{align}
If in \eqref{m} the weights $\omega_{\h}(\T_i, \T_k)$ are given by \eqref{pesos0}, $\bm{T}_k$ are the covariates, and each $Y_{k}$ denotes the response variable, then under assumptions A1--A6, and by Remark \ref{hipAnne}, we directly apply Observation~4 of \cite{DaboNiang2016}, obtaining the same orders as in \eqref{30graya}, namely,
\begin{equation}\label{rraya}
\max_{1\le i \le \hat{n}} 
\left| \hope{Y_i |\T_i} -\sum_{k=1}^{\hat{n}} \omega_{\h}(\T_i, \T_k)Y_k \right|
= O(h_{2\nn}) 
+ O\!\left(\frac{\log \hat{n}}{\hat{n} h^d_{1\nn}h^k_{2\nn}} \right).
\end{equation}
In \eqref{eta}, using \eqref{33eta}, for each $j$-th component of the vector $\sum_{k=1}^{\hat{n}} \omega_{\h}(\T_i, \T_k)\big(\bm{X}_{k}-\hope{\bm{X}_k |\T_k}\big)$, we have
\begin{align}
\label{eta2}
\bigg(\sum_{k=1}^{\hat{n}} \omega_{\h}(\T_i, \T_k)\big(\bm{X}_{k}-\hope{\bm{X}_k |\T_k}\big)\bigg)_j 
= O\left( ( h^d_{1\nn}h^k_{2\nn})^{-1} \hat{n}^{-1/2 + 1/r} \log \hat{n} \right)
\quad \text{a.s.}
\end{align}
In \eqref{gPelito2}, using \eqref{35}, we obtain
\begin{align}
\label{gpelito3}
\bigg(\hope{\bm{X}_i |\T_i} - \sum_{k=1}^{\hat{n}} \omega_{\h}(\T_i, \T_k)\hope{\bm{X}_k |\T_k}\bigg)_j
= O(h_{2\nn})
+ O\left( (h^d_{1\nn}h^k_{2\nn})^{-1} \hat{n}^{-1/2 + 1/r} \log \hat{n} \right)
\quad \text{a.s.}
\end{align}
Finally, combining \eqref{32}, \eqref{gpelito3}, \eqref{eta2}, and \eqref{rraya}, we obtain
\begin{align}
\max_i |\tilde r_\textbf{h}(\T_i)|
&= O(h_{2\nn})
+ O\left( \sqrt{\frac{\log \hat{n}}{n h^d_{1\nn}h^k_{2\nn}}} \right)
+ O\left(  (h^d_{1\nn}h^k_{2\nn})^{-1} \hat{n}^{-1/2 + 1/r} \log \hat{n} \right)\\
\label{34}
&= O(h_{2\nn})
+ O\left(  (h^d_{1\nn}h^k_{2\nn})^{-1} \hat{n}^{-1/2 + 1/r} \log \hat{n} \right)
\quad \text{a.s.}
\end{align}

We next analyze the nine terms corresponding to the $j$-th components of $\bm{S}_{n1}$, $\bm{S}_{n2}$ and $\bm{S}_{n3}$.

For $S_{n1,j1}$, using \eqref{34} and \eqref{35}, we obtain
\begin{align}
\hat{n}^{-\frac{1}{2}}S_{n1,j1}
&= \hat{n}^{-\frac{1}{2}}\sum_{i=1}^{\hat{n}} \tilde{g}_{j}(\T_i)\,\tilde{r}_h(\T_i) \leq  \hat{n}^{\frac{1}{2}} \max_i |\tilde{g}_{j}(\T_i)| \max_i |\tilde{r}(\T_i)| 
= O\left((\hat{n} h_{2\nn}^{4})^{\frac{1}{2}}\right) 
+ O\!\left(\bigg(\frac{\log \hat{n}}{ h^d_{1\nn}h^k_{2\nn} \hat{n}^{\frac{1}{4}-\frac{1}{r}}}\bigg)^2\right)
\quad \text{a.s.}
\end{align}
For $S_{n1,j2}$, applying Lemma 3 from \cite{Aneiros-Perez2008Nonparametric} with \( a_{ik} = \tilde{r}_\h(T_k)\), \(a_{\nn} = h_{2\nn} +  (h^d_{1\nn}h^k_{2\nn})^{-1} \hat{n}^{-\frac{1}{2} + \frac{1}{r}} \log \hat{n}\) (see \eqref{34}), \( V_k = \eta_{kj} \), and \( 0.5 < \gamma < 1 - 9/(4b) \), we obtain
\[
\hat{n}^{-\frac{1}{2}}S_{n1,j2}=  \hat{n}^{-\frac{1}{2}}\sum_{k=1}^{\hat{n}} \tilde{r}(T_k)\,\eta_{kj}= O\!\left(( h^4_{2\nn} \hat{n})^{\frac{1}{4}}  \frac{{\log \hat{n} }}{\hat{n}^{\frac{1}{4}-\frac{1}{r}}}
+  \bigg(\frac{\log \hat{n}}{ (h^d_{1\nn}h^k_{2\nn} )^{\frac{1}{2}}\hat{n}^{\frac{1}{4}-\frac{1}{r}}}\bigg)^2 \right)
\quad \text{a.s.}
\]
For $S_{n1,j3}$, using \eqref{33eta} and \eqref{34}, we obtain
\begin{align}
\hat{n}^{-\frac{1}{2}}S_{n1,j3} 
&= \hat{n}^{-\frac{1}{2}}\sum_{i=1}^{\hat{n}} \tilde{r}(\T_i)
\left( \sum_{k=1}^{\hat{n}} \omega_{\h}(\T_i, \T_k) \eta_{kj} \right)\\
&\leq \hat{n}^{\frac{1}{2}}\max_i |\tilde{r}(\T_i)| \max_i \left| \sum_{k=1}^{\hat{n}} \omega_{\h}(\T_i, \T_k)\eta_{kj} \right| 
=  O\bigg(( h^4_{2\nn} \hat{n})^{\frac{1}{4}}  \frac{\log \hat{n}}{h^d_{1\nn}h^k_{2\nn}\hat{n}^{\frac{1}{4}-\frac{1}{r}} } 
+ \bigg(\frac{\log \hat{n}}{h^d_{1\nn}h^k_{2\nn}\hat{n}^{\frac{1}{4}-\frac{1}{r}}} \bigg)^{2}   \bigg) 
\quad \text{a.s.}
\end{align}
Analogously, using \eqref{32} and \eqref{35}, we obtain for $S_{n2,j1}$
\begin{align}
\hat{n}^{-\frac{1}{2}}S_{n2,j1} 
&= \hat{n}^{-\frac{1}{2}}\sum_{i=1}^{\hat{n}} \tilde{g}_{j}(\T_i)
\left( \sum_{k=1}^{\hat{n}} \omega_{\h}(\T_i, \T_k) \varepsilon_k \right)\\ 
&\leq \hat{n}^{\frac{1}{2}}\max_i |\tilde{g}_{j}(\T_i)| \max_i \left| \sum_{k=1}^{\hat{n}} \omega_{\h}(\T_i, \T_k)\varepsilon_k \right|
=  O\bigg(( h^4_{2\nn} \hat{n})^{\frac{1}{4}}  \frac{\log \hat{n}}{h^d_{1\nn}h^k_{2\nn}\hat{n}^{\frac{1}{4}-\frac{1}{r}} } 
+ \bigg(\frac{\log \hat{n}}{h^d_{1\nn}h^k_{2\nn}\hat{n}^{\frac{1}{4}-\frac{1}{r}}} \bigg)^{2}   \bigg) 
\quad \text{a.s.}
\end{align}
Applying Lemma 3 from \cite{Aneiros-Perez2008Nonparametric} with \( a_{ik} = \sum_{l=1}^{\hat{n}}  \omega_{\h}(\T_k, \T_l)\varepsilon_l \),  
\( a_{\nn} =  (h^d_{1\nn}h^k_{2\nn})^{-1} \hat{n}^{-\frac{1}{2}+\frac{1}{r}} \log \hat{n} \),
\( V_k = \eta_{kj} \) and \( 0.5 < \gamma < 1 - 9/(4b) \), we obtain
\begin{align}
\hat{n}^{-\frac{1}{2}}S_{n2,j2}
&= \hat{n}^{-\frac{1}{2}}\sum_{k=1}^{\hat{n}}
\left( \sum_{l=1}^{\hat{n}} \omega_{\h}(\T_k, \T_l)\varepsilon_l \right) \eta_{kj}
= O\!\left(
\bigg(\frac{\log \hat{n}}{ (h^d_{1\nn}h^k_{2\nn} )^{\frac{1}{2}}\hat{n}^{\frac{1}{4}-\frac{1}{r}}}\bigg)^2 \right)
\quad \text{a.s.}
\end{align}
Similarly, using \eqref{32} and \eqref{33eta}, we obtain for $S_{n2,j3}$
\begin{align}
\hat{n}^{-\frac{1}{2}}S_{n2,j3}
&= \hat{n}^{-\frac{1}{2}}\sum_{k=1}^{\hat{n}} \left( \sum_{l=1}^{\hat{n}} \omega_{\h}(\T_k, \T_l)\eta_{lj} \right)
\left( \sum_{l=1}^{\hat{n}} \omega_{\h}(\T_k, \T_l)\varepsilon_l \right)\\
&\leq \hat{n}^{\frac{1}{2}} 
\max_i \left| \sum_{k=1}^{\hat{n}} \omega_{\h}(\T_i, \T_k)\eta_{kj} \right|
\max_i \left| \sum_{l=1}^{\hat{n}} \omega_{\h}(\T_i, \T_l)\varepsilon_l \right| 
= O\!\left(
\bigg(\frac{\log \hat{n}}{ (h^d_{1\nn}h^k_{2\nn} )^{\frac{1}{2}}\hat{n}^{\frac{1}{4}-\frac{1}{r}}}\bigg)^2 \right)
\quad \text{a.s.}
\end{align}
Applying Lemma 3 from \citet{Aneiros-Perez2008Nonparametric} with \( a_{ik} = \tilde{g}_{j}(\T_k)\),  
\( a_{\nn} = h_{2\nn} +  (h^d_{1\nn}h^k_{2\nn})^{-1} \hat{n}^{-\frac{1}{2} + \frac{1}{r}} \log \hat{n} \) (see \ref{35}),  
\( V_k = \varepsilon_{k} \), and \( 0.5 < \gamma < 1 - 9/(4b) \), we obtain
\begin{align}
\hat{n}^{-\frac{1}{2}}S_{n3,j1}
&= \hat{n}^{-\frac{1}{2}}\sum_{k=1}^{\hat{n}} \tilde{g}_{j}(\T_k)\varepsilon_k\\
&= O\!\left( h_{2\nn} \hat{n}^{\frac{1}{r}} \log \hat{n}
+  (h^d_{1\nn}h^k_{2\nn})^{-1} \hat{n}^{\frac{2}{r}-\frac12}\log^2 \hat{n} \right)\label{42}
= O\!\left( (h^4_{2\nn} \hat{n})^{\frac{1}{4}}  \frac{{\log \hat{n} }}{\hat{n}^{\frac{1}{4}-\frac{1}{r}}}
+  \bigg(\frac{\log \hat{n}}{ (h^d_{1\nn}h^k_{2\nn} )^{\frac{1}{2}}\hat{n}^{\frac{1}{4}-\frac{1}{r}}}\bigg)^2 \right)
\quad \text{a.s.}
\end{align}
Applying Lemma 3 from \cite{Aneiros-Perez2008Nonparametric} with  
\( a_{ik} = \sum_{l=1}^{\hat{n}} \omega_{\h}(\T_k, \T_l)\eta_{lj} \),  
\( a_{\nn} =  (h^d_{1\nn}h^k_{2\nn})^{-1} \hat{n}^{-\frac{1}{2}+\frac{1}{r}} \log \hat{n} \) (see \eqref{33eta}),  
\( V_k = \varepsilon_{k} \), we obtain $
\hat{n}^{-\frac{1}{2}}S_{n3,j3}
= \hat{n}^{-\frac{1}{2}}\sum_{k=1}^{\hat{n}} \left( \sum_{l=1}^{\hat{n}} \omega_{\h}(\T_k, \T_l)\eta_{lj} \right) \varepsilon_k
= O\!\left(
\bigg(\frac{\log \hat{n}}{ (h^d_{1\nn}h^k_{2\nn} )^{\frac{1}{2}}\hat{n}^{\frac{1}{4}-\frac{1}{r}}}\bigg)^2 \right)
\quad \text{a.s.}$ 

To show that the obtained orders converge to zero, we use the assumptions \eqref{H0}.
From that, it follows that $\hat{n}^{-\frac12} S_{n1,j1}$, 
$\hat{n}^{-\frac12} S_{n1,j3}$ and $\hat{n}^{-\frac12} S_{n2,j1}$ converge to zero.
Moreover, since $h_{1\nn}, h_{2\nn} \to 0$, for sufficiently large $\hat{n}$ we have
\[\label{hh}
\frac{\log \hat{n}}{ (h^d_{1\nn}h^k_{2\nn} )^{\frac{1}{2}}\hat{n}^{\frac{1}{4}-\frac{1}{r}}}
\le 
\frac{\log \hat{n}}{ (h^d_{1\nn}h^k_{2\nn} )\hat{n}^{\frac{1}{4}-\frac{1}{r}}}.
\]
Hence, by \eqref{H0} and \eqref{hh}, the terms $\hat{n}^{-\frac12} S_{n2,j2}$, 
$\hat{n}^{-\frac12} S_{n2,j3}$ and $\hat{n}^{-\frac12} S_{n3,j3}$ converge to zero. Finally, since $r>4$ (by A3), we have
\(\label{nn}
{\log\hat{n}}/{\hat{n}^{\frac14-\frac1r}} \to 0,
\)
which, together with \eqref{H0} and \eqref{hh}, implies that $\hat{n}^{-\frac12} S_{n1,j2}$, 
and $\hat{n}^{-\frac12} S_{n3,j1}$ converge to zero. Combining the above results, we obtain
\[
S_{n1,j1}+S_{n1,j2}+S_{n1,j3}-S_{n2,j1}-S_{n2,j2}-S_{n2,j3}+S_{n3,j1}+S_{n3,j3} = o(\hat{n}^{\frac{1}{2}}).
\]
Therefore,
\begin{equation}\label{44}
\bm{S}_{n1} - \bm{S}_{n2} + \bm{S}_{n3}
= \sum_{i=1}^{\hat{n}} \bm{\eta}_i \varepsilon_i + o(\hat{n}^\frac{1}{2})
\quad \text{a.s.}
\end{equation}

In the following, we aim to show that
\[\label{45}
\hat{n}^{-1/2} S_{n3,j2}
=
\hat{n}^{-1/2} \sum_{i=1}^{\hat{n}} \bm{\eta}_i \varepsilon_i
\;\xrightarrow{d}\;
\mathcal{N}(0,\mathbf{C}),
\]
where 
\(
\mathbf{C}
=
\lim_{\hat{n}\to\infty}
\hat{n}^{-1}
\mathbb{E}\!\left[\bm{\eta}^{T} \textbf{V}_\varepsilon\bm{\eta}\right],
\)
as defined in Assumption A\ref{covas}.

To achieve this, for each $j$, we apply the main Theorem in \cite{Bolthausen1982} to 
\begin{equation}\label{etaepsilonj}
\left\{ {\eta}_{ij} \varepsilon_i \right\}.
\end{equation}
To this end, it suffices to verify the assumptions of the above theorem for \eqref{etaepsilonj}. 
First, the random field $\{\bm{\eta}_i \varepsilon_i\}$ has zero mean, is stationary, and satisfies the $\alpha$-mixing condition. Moreover, there exists $\delta>0$ such that
\begin{equation}\label{normacot}
\mathbb{E}\!\left[\|\bm{\eta}_i \varepsilon_i\|^{2+\delta}\right] < \infty.
\end{equation}
These properties follow directly from Lemma \ref{lemaultimo} and imply that they hold for each coordinate.

Furthermore, by the definition of $\alpha(m)$ in \eqref{alfa} (which does not impose cardinality restrictions on the index sets) and Assumption A\ref{mixing}, for all $\kappa, \ell \in \mathbb{N}$, $\alpha_{\kappa,\ell}(m) \le \alpha(m) \le C m^{-b}. $ In particular, $\alpha_{1,\infty}(m) \le \alpha(m)$, and hence
\(
\alpha_{1,\infty}(m) = o(m^{-d}).
\)
Moreover,
\(
\alpha_{1,1}(m)^{\frac{\delta}{2+\delta}}
\le C\, m^{-b\frac{\delta}{2+\delta}}.
\)
Thus,
\(
\sum_{m=1}^{\infty} m^{d-1}\,\alpha_{1,1}(m)^{\frac{\delta}{2+\delta}}
\le
C \sum_{m=1}^{\infty} m^{d-1 - b\frac{\delta}{2+\delta}},
\)
which is finite whenever $b\frac{\delta}{2+\delta} > d$. Since $b>4.5d$, one can choose $\delta>0$ (e.g., any $\delta<r-2$, with $r>4$ from the moment assumptions) such that this condition holds. Consequently,
\(
\sum_{m=1}^{\infty} m^{d-1}\,\alpha_{1,1}(m)^{\frac{\delta}{2+\delta}} < \infty.
\)
This verifies condition (b) of the main Theorem in \cite{Bolthausen1982}, while \eqref{normacot} ensures the required moment condition.

We now identify the limiting variance. For each $j=1,\dots,p$, define
\(
Z_\s = \eta_{\s j}\varepsilon_\s,
\)
for $ \s \in \mathbb{Z}^d. $ As $\{\s_1,\dots,\s_{\hat n}\}$ is an enumeration of $\mathcal{D}_\nn$, then $Z_i = Z_{\s_i}$ and
\(
\sum_{i=1}^{\hat n} Z_i = \sum_{\s \in \mathcal{D}_\nn} Z_\s.
\)
By stationarity,
\[
\mathrm{Var}\Big(\sum_{\s \in \mathcal{D}_\n} Z_\s\Big)
=
\sum_{\textbf{v} \in \mathbb{Z}^d} N_{\hat n}(\textbf{v})\,\mathrm{Cov}(Z_0, Z_\textbf{v}),
\]
where $N_{\hat{n}}(\textbf{v}) = \#\{\s\in \mathcal{D}_\nn: \s+\textbf{v} \in \mathcal{D}_\nn\}$ denotes the number of pairs at distance $\textbf{v}$. Dividing by $\hat{n}$ yields
\[
\hat{n}^{-1}
\mathrm{Var}\Big(\sum_{i=1}^{\hat{n}} Z_i\Big)
=
\sum_{\textbf{v} \in \mathbb{Z}^d}
\frac{N_{\hat{n}}(\textbf{v})}{\hat{n}}\,\mathrm{Cov}(Z_0,Z_\textbf{v}),
\]
which follows from rewriting the double sum of covariances in terms of spatial lags. Since $\frac{N_{\hat{n}}(\textbf{v})}{\hat{n}} \to 1$ for each fixed $\textbf{v}$, and the series $\sum_{\textbf{v}} |\mathrm{Cov}(Z_0,Z_\textbf{v})|$ is finite, the dominated convergence theorem implies
\[
\hat{n}^{-1}
\mathrm{Var}\Big(\sum_{i=1}^{\hat{n}} Z_i\Big)
\;\longrightarrow\;
\sum_{\textbf{v} \in \mathbb{Z}^d} \mathrm{Cov}(Z_0,Z_\textbf{v}).
\]

Hence, the limiting variance coincides with the corresponding diagonal entry $\mathbf{C}_{jj}$ of $\mathbf{C}$, and
\[
\hat{n}^{-1/2} \sum_{i=1}^{\hat{n}} \eta_{ij} \varepsilon_i 
\xrightarrow{d} \mathcal{N}(0, \mathbf{C}_{jj}).
\]

To establish the joint convergence
$\hat{n}^{-1/2} S_{n3,j2} \xrightarrow{d} \mathcal{N}(0, \mathbf{C})$,
we apply the Cramér-Wold device. It suffices to show that for every fixed
$\bm{t} \in \mathbb{R}^k \setminus \{0\}$,
\(
\hat{n}^{-1/2} \sum_{i=1}^{\hat{n}} \bm{t}^T \bm{\eta}_i \varepsilon_i 
\xrightarrow{d} \mathcal{N}(0, \bm{t}^T \mathbf{C} \bm{t}).
\)
Since $\tB^T \eta_i \varepsilon_i$ is scalar, this reduces to a scalar problem. 
The moment condition $\mathbb{E}|\bm{t}^T \bm{\eta}_i \varepsilon_i|^{2+\delta} < \infty$ follows from Lemma \ref{lemaultimo}, together with the bound $|\bm{t}^T \bm{\eta}_i \varepsilon_i| \leq \|\bm{t}\| \|\bm{\eta}_i\||\varepsilon_i|$. 
The mixing coefficients are properties of the underlying random field and are therefore unaffected by the projection. Finally, since $\mathbf{C}$ is positive definite by assumption, the limiting variance $\bm{t}^T \mathbf{C} \bm{t} > 0$ for all 
$\bm{t} \neq 0$. Hence the scalar central limit theorem applies, and the result follows by the Cramér-Wold device.

Combining \eqref{44} and \eqref{45}, we obtain
\(\label{lim1_paper}
\hat{n}^{-1/2}(\bm S_{n1}-\bm S_{n2}+\bm S_{n3}) \xrightarrow{d} N(0,\mathbf{C}).
\)
Moreover, from Lemma \ref{lema5},
\begin{equation}\label{B}
\hat{n}^{-1}\tilde{\mathbb{X}}^{T}\tilde{\mathbb{X}} \xrightarrow{\text{a.s.}} \mathbf{B}.
\end{equation}
Therefore, applying Slutsky's theorem to \eqref{31}, it follows that
\[
\sqrt{\hat{n}}\left(\hat{\betaB}_\mathbf{h} - \betaB \right)
=
(\hat{n}^{-1}\tilde{\mathbb{X}}^{T}\tilde{\mathbb{X}})^{-1}
\hat{n}^{-1/2}(\bm S_{n1}-\bm S_{n2}+\bm S_{n3})
\xrightarrow{d}
N\bigl(0,\mathbf{B}^{-1}\mathbf{C}\mathbf{B}^{-1}\bigr),
\]
which establishes the first part of the theorem.

We now address \eqref{convergenciabeta},
where $a_{jj} = \mathbf{A}_{jj}$ and $\mathbf{A} = \mathbf{B}^{-1}\mathbf{C}\mathbf{B}^{-1}$, we consider the decomposition in \eqref{31}. From \eqref{B} and \eqref{44}, it follows that
\begin{align}
\hat{\betaB}_\mathbf{h} - \betaB
&=
\left(\hat{n}^{-1} \tilde{\mathbb{X}}^{T} \tilde{\mathbb{X}} \right)^{-1}
\hat{n}^{-1} (S_{n1} - S_{n2} + S_{n3}) \nonumber\\\label{46_paper}
&=
\left( \mathbf{B}^{-1} + o(1) \right)
\left( \hat{n}^{-1} \sum_{i=1}^{\hat{n}} \bm{\eta}_i \varepsilon_i + o\big(\hat{n}^{-1/2}\big) \right)
\quad \text{a.s.}
\end{align}
Focusing on the $j$-th component and denoting by $\bm{b}_{j\cdot}$ the $j$-th row of $\mathbf{B}^{-1}$, it suffices to analyze
\[
\left(\frac{\hat{n}}{2\log \log \hat{n}}\right)^{1/2}
\hat{n}^{-1}\bm{b}_{j\cdot}\sum_{i=1}^{\hat{n}} \bm{\eta}_i \varepsilon_i
=
\left(\frac{1}{2\hat{n}\log \log \hat{n}}\right)^{-1/2}
\bm{b}_{j\cdot}\sum_{i=1}^{\hat{n}} \bm{\eta}_i \varepsilon_i.
\]
To this end, we apply Theorem 5 from \cite{Oodaira1971} with $V_i = \bm{b}_{j\cdot}\bm{\eta}_i \varepsilon_i$. By Lemma \ref{lemaultimo} (under A\ref{fs2}, A\ref{ryg}, A\ref{mixing} and A\ref{epsilonEta}), and since $\bm{b}_{j\cdot}$ is nonrandom and does not depend on $i$, the sequence $\{V_i\}$ is strictly stationary, $\alpha$-mixing, has zero mean, and satisfies $\mathbb{E}|V_i|^{2+\delta}<\infty$.
On the other hand, by Lemma \ref{lemamixing}, we have
\[
\sum_{n=1}^\infty \alpha(\hat{n})^{\frac{\delta'}{2+\delta'}}
\le 
\sum_{n=1}^\infty \left(\hat{n}^{-\frac{b}{d}}\right)^{\frac{\delta'}{2+\delta'}}.
\]
Moreover, under assumptions~A\ref{mixing}, $b>4.5\,d$, and choosing
$\frac{2}{b-d}<\delta'<\delta<r-2$, the series is convergent.

Let $S_n = \sum_{i=1}^{\hat{n}} V_i$. It remains to verify that the asymptotic variance of $S_n$ is positive.
Since $\bm{\eta}$ and $\bm{\varepsilon}$ are independent (A\ref{epsilonEta}) and $\hope{\bm{\varepsilon}}=\bm 0$, we have $\Var{S_n}
= \Var{\sum_{i=1}^{\hat{n}} \bm{b}_{j\cdot}\bm{\eta}_i\varepsilon_i}
= \bm{b}_{j\cdot}\hope{\bm{\eta}^T V_\varepsilon \bm{\eta}}\bm{b}_{j\cdot}^T.$ Dividing by $\hat{n}$, we obtain $\frac{1}{\hat{n}}\Var{S_n}
=
\bm{b}_{j\cdot}
\left(
\frac{1}{\hat{n}}\hope{\bm{\eta}^T V_{\varepsilon} \bm{\eta}}
\right)
\bm{b}_{j\cdot}^T. $ By hypothesis H7, the limit
\(
\mathbf{C}
=
\lim_{\hat{n}\to\infty}
\frac{1}{\hat{n}}\hope{\bm\eta^T V_\varepsilon \bm\eta}
\)
exists and is positive definite. Consequently,
\(
\lim_{\hat{n}\to\infty}
\frac{1}{\hat{n}}\Var{S_n}
=
\bm{b}_{j\cdot}\mathbf{C}\bm{b}_{j\cdot}^T > 0,
\)
since $\bm{b}_{j\cdot}\neq 0$. Hence, all the conditions of Theorem 5 from \cite{Oodaira1971} are satisfied.

Moreover, using \eqref{varSuma_paper},
\(
s_n^2 \defi \hope{S_{\hat{n}}^2}
= \Var{\sum_{i=1}^{\hat{n}} \bm{b}_{j\cdot} \bm{\eta}_i \varepsilon_i}
= \hat{n}\, a_{jj} \bigl(1 + o(1)\bigr).
\)
Therefore, by Theorem 5 from \cite{Oodaira1971},
\begin{equation}\label{47_paper}
\limsup_{\hat{n} \to \infty} 
\left( \frac{1}{2\hat{n} \log \log \hat{n}} \right)^{1/2}
\left| \sum_{i=1}^{\hat{n}}\bm{b}_{j\cdot} \bm{\eta}_i \varepsilon_i \right|
= (a_{jj})^{1/2}
\quad \text{a.s.}
\end{equation}
Combining \eqref{46_paper} and \eqref{47_paper} completes the proof of the second part of the theorem.
\end{proof}
\begin{proof}[Proof of Theorem \ref{r}]
    Starting from \eqref{BLUPnp} and \eqref{modeloespacial}, we obtain
\begin{equation}
\hat{r}(\tB)
= \hat{r}^*_{\textbf{h}}(\tB)
- \sum_{i=1}^{\hat{n}} \omega_{\textbf{h}}(\tB, \T_i)\bm{X}_i^{\top}(\hat{\betaB}_{\textbf{h}} - \betaB),
\label{eq:decomposition_compact}
\end{equation}
where
\(
\hat{r}^*_{\textbf{h}}(\tB)
= \sum_{i=1}^{\hat{n}} \omega_{\textbf{h}}(\tB, \T_i)\big(r(\T_i) + \varepsilon_i\big).
\)

Under assumptions H1--H5 and by Remark~\ref{hipAnne}, together with Observation~4 of Theorem 3.1 in \cite{DaboNiang2016}, it holds that
\begin{equation}\label{parte1}
\sup_{\tB \in \mathcal{C}}
\big| \hat{r}^*_{\textbf{h}}(\tB) - r(\tB) \big|
= O(h_{2\nn})
+ O\!\left( \sqrt{\frac{\log \hat{n}}{\hat{n} h_{1\nn}^d h_{2\nn}^k}} \right)
\quad \text{a.s.}
\label{eq:rstar_rate}
\end{equation}
Moreover, from \eqref{convergenciabeta}, for each $j=1,\dots,p$,
\[\label{parte2}
\big| \hat{\betaB}_{\textbf{h}j} - \betaB_j \big|
= O\!\left(\sqrt{\frac{\log \log \hat{n}}{\hat{n}}}\right)
\quad \text{a.s.}
\]
Since $p$ is fixed, this implies
\(
\|\hat{\betaB}_{\textbf{h}} - \betaB\|
= O\!\left(\sqrt{{\log \log \hat{n}}/{\hat{n}}}\right)
\) almost sure too.

Therefore, using \eqref{eq:decomposition_compact}, we obtain
\begin{align}
\sup_{\tB \in \mathcal{C}} \big| \hat{r}(\tB) - r(\tB) \big|
&\leq
\sup_{\tB \in \mathcal{C}} \big| \hat{r}^*_{\textbf{h}}(\tB) - r(\tB) \big|
+ \sup_{\tB \in \mathcal{C}} \left|
\sum_{i=1}^{\hat{n}} \omega_{\textbf{h}}(\tB, \T_i)\bm{X}_i^{\top}
(\hat{\betaB}_{\textbf{h}} - \betaB)
\right| \notag \\
&= O(h_{2\nn})
+ O\!\left( \sqrt{\frac{\log \hat{n}}{\hat{n} h_{1\nn}^d h_{2\nn}^k}} \right)
+ O\!\left(\sqrt{\frac{\log \log \hat{n}}{\hat{n}}}\right)
\quad \text{a.s.}
\label{eq:final_uniform_rate}
\end{align}
Since the weights are nonnegative (A1) and satisfy $
\sum_{i=1}^{\hat{n}} \omega_{\textbf{h}}(\tB,\T_i)=1$ it follows, from the triangle inequality that, $\left\|
\sum_{i=1}^{\hat{n}} \omega_{\textbf{h}}(\tB,\T_i) \bm{X}_i
\right\|
\le
\sum_{i=1}^{\hat{n}} \omega_{\textbf{h}}(\tB,\T_i)\|\bm{X}_i\|. $From the definition of the weights in \eqref{pesosFraccion}, 
and since the kernels have compact support, only those indices such that
\(
d(\s,\s_i)\le C h_{1\nn},\)
and
\(d_m(\tB,\T_i)\le C h_{2\nn},
\)
contribute to the sum. Let $\mathcal{I}_{\nn}(\tB)$ denote this active set of indices. Then,
\(
|\mathcal{I}_{\nn}(\tB)| = O\big(\hat{n} h_{1\nn}^d h_{2\nn}^k\big),
\)
almost sure.

Moreover, by Lemma~\ref{lema4}, we have \eqref{cotaPeso} we have that $
\sum_{i=1}^{\hat{n}} \omega_{\textbf{h}}(\tB,\T_i)\|\bm{X}_i\|
=
O\!\left(\frac{1}{\hat{n} h_{1\nn}^d h_{2\nn}^k}\right)
\sum_{i \in \mathcal{I}_{\nn}(\tB)} \|\bm{X}_i\|
\quad \text{a.s.} $ . Since $E\|\bm{X}\| < \infty$ and the sequence is $\alpha$-mixing, the strong law of large numbers (Theorem 2(b) from \cite{Devroye2010}), applied under the same conditions as in \eqref{30graya}, yields $
\frac{1}{|\mathcal{I}_{\nn}(\tB)|}
\sum_{i \in \mathcal{I}_{\nn}(\tB)} \|\bm{X}_i\|
\to
E\|\bm{X}\|
\quad \text{a.s.} $ uniformly in $\tB \in \mathcal{C}$, since the index sets $\mathcal{I}_{\nn}(\tB)$ correspond to local neighborhoods of comparable size over a compact domain under a regular grid design. Hence,
\(
\sum_{i \in \mathcal{I}_{\nn}(\tB)} \|\bm{X}_i\|
=
O\big(|\mathcal{I}_{\nn}(\tB)|\big),
\)
almost sure and  uniformly in $ \tB. $

Recalling that $|\mathcal{I}_{\nn}(\tB)| = O(\hat{n} h_{1\nn}^d h_{2\nn}^k)$ a.s., we obtain $\sup_{\tB \in \mathcal{C}}
\sum_{i=1}^{\hat{n}} \omega_{\textbf{h}}(\tB,\T_i)\|\bm{X}_i\|
=
O(1)
\quad \text{a.s.} $ Combining the previous bound with the convergence rate of $\hat{\betaB}_{\textbf{h}}$ in \eqref{parte2} and the uniform rate in \eqref{parte1}, we obtain
\begin{align}
\sup_{\tB \in \mathcal{C}}\big|\hat{r}_{\textbf{h}}(\tB) - r(\tB)\big|
&\le
\sup_{\tB \in \mathcal{C}}\big|\hat{r}^*_{\textbf{h}}(\tB) - r(\tB)\big|
+ \sup_{\tB \in \mathcal{C}} \sum_{i=1}^{\hat{n}}
\omega_{\textbf{h}}(\tB, \T_i)\|\bm{X}_i\|
\, \|\hat{\betaB}_{\textbf{h}} - \betaB\| \notag \\
&= O(h_{2\nn})
+ O\!\left( \sqrt{\frac{\log \hat{n}}{\hat{n} h_{1\nn}^d h_{2\nn}^k}} \right)
+ O\!\left(\sqrt{\frac{\log \log \hat{n}}{\hat{n}}}\right)
\quad \text{a.s.}
\label{eq:final_rate_full}
\end{align}
Since $
\sqrt{\frac{\log \log \hat{n}}{\hat{n}}}
= o\!\left(
\sqrt{\frac{\log \hat{n}}{\hat{n} h_{1\nn}^d h_{2\nn}^k}}
\right), $ under the imposed assumptions, it follows that
\begin{equation}
\sup_{\tB \in \mathcal{C}}\big|\hat{r}_{\textbf{h}}(\tB) - r(\tB)\big|
=
O(h_{2\nn})
+ O\!\left( \sqrt{\frac{\log \hat{n}}{\hat{n} h_{1\nn}^d h_{2\nn}^k}} \right)
\quad \text{a.s.}
\label{eq:final_rate_simplified}
\end{equation}

This completes the proof.

\end{proof}


\begin{thebibliography}{11}

	
	
	\bibitem[\protect\citeauthoryear{Aneiros-Pérez and Vieu}{Aneiros-Pérez and Vieu}{2008}]{Aneiros-Perez2008Nonparametric}
	Aneiros-Pérez, G. and P.~Vieu (2008).
	\newblock Nonparametric time series prediction: A semi-functional partial linear modeling.
	\newblock {\em J. of Multivariate Analysis\/}~{\em 99}, 834--857.
	
	\bibitem[\protect\citeauthoryear{Basile, Durb{\'a}n, M{\'\i}nguez, Montero, and Mur}{Basile et~al.}{2014}]{basile2014modeling}
	Basile, R., M.~Durb{\'a}n, R.~M{\'\i}nguez, J.~M. Montero, and J.~Mur (2014).
	\newblock Modeling regional economic dynamics: Spatial dependence, spatial heterogeneity and nonlinearities.
	\newblock {\em J. of Econ. Dynamics and Control\/}~{\em 48}, 229--245.
	
	\bibitem[\protect\citeauthoryear{Bauman, Drouet, Fortin, and Dray}{Bauman et~al.}{2018}]{bauman2018optimizing}
	Bauman, D., T.~Drouet, M.-J. Fortin, and S.~Dray (2018).
	\newblock Optimizing the choice of a spatial weighting matrix in eigenvector-based methods.
	\newblock {\em Ecology\/}~{\em 99\/}(10), 2159--2166.
	
	\bibitem[\protect\citeauthoryear{Bolthausen}{Bolthausen}{1982}]{Bolthausen1982}
	Bolthausen, E. (1982).
	\newblock {On the Central Limit Theorem for Stationary Mixing Random Fields}.
	\newblock {\em The Annals of Probability\/}~{\em 10\/}(4), 1047 -- 1050.
	
	\bibitem[\protect\citeauthoryear{Cerqueti, Maranzano, and Mattera}{Cerqueti et~al.}{2025}]{cerqueti2025spatially}
	Cerqueti, R., P.~Maranzano, and R.~Mattera (2025).
	\newblock Spatially-clustered spatial autoregressive models with application to agricultural market concentration in europe.
	\newblock {\em J. of Agricultural, Biological and Environmental Stat.\/}, 1--35.
	
	\bibitem[\protect\citeauthoryear{Dabo-Niang, Ternynck, and Yao}{Dabo-Niang et~al.}{2016}]{DaboNiang2016}
	Dabo-Niang, S., C.~Ternynck, and A.-F. Yao (2016).
	\newblock Nonparametric prediction of spatial multivariate data.
	\newblock {\em J. of Nonparametric Stat.\/}~{\em 28\/}(2), 428--458.
	
	\bibitem[\protect\citeauthoryear{Duncan, White, and Mengersen}{Duncan et~al.}{2017}]{duncan2017spatial}
	Duncan, E.~W., N.~M. White, and K.~Mengersen (2017).
	\newblock Spatial smoothing in bayesian models: a comparison of weights matrix specifications and their impact on inference.
	\newblock {\em Int. J. of health geographics\/}~{\em 16\/}(1), 47.
	
	\bibitem[\protect\citeauthoryear{Fouedjio and Arya}{Fouedjio and Arya}{2024}]{Fouedjio2024_GeoML}
	Fouedjio, F. and E.~Arya (2024).
	\newblock Locally varying geostatistical machine learning for spatial prediction.
	\newblock {\em Artificial Intelligence in Geosciences\/}~{\em 5}, 100081.
	
	\bibitem[\protect\citeauthoryear{Gao, Lu, and Tjøstheim}{Gao et~al.}{2006}]{Gao2006}
	Gao, J., Z.~Lu, and D.~Tjøstheim (2006, June).
	\newblock Estimation in semiparametric spatial regression.
	\newblock {\em The Annals of Stat.\/}~{\em 34\/}(3).
	
	\bibitem[\protect\citeauthoryear{Garc{\'\i}a~Arancibia, Llop~Orzan, and Lovatto}{Garc{\'\i}a~Arancibia et~al.}{2023}]{GarciaArancibia2023NPS}
	Garc{\'\i}a~Arancibia, R., P.~N. Llop~Orzan, and M.~G. Lovatto (2023).
	\newblock Nonparametric prediction for univariate spatial data: methods and applications.
	\newblock {\em Papers in Regional Science\/}.
	
	\bibitem[\protect\citeauthoryear{Hoshino}{Hoshino}{2018}]{hoshino2018semiparametric}
	Hoshino, T. (2018).
	\newblock Semiparametric spatial autoregressive models with endogenous regressors: With an application to crime data.
	\newblock {\em J. of Business \& Econ. Stat.\/}~{\em 36\/}(1), 160--172.
	
	\bibitem[\protect\citeauthoryear{Jenish}{Jenish}{2014}]{Jenish2014}
	Jenish, N. (2014, December).
	\newblock Spatial semiparametric model with endogenous regressors.
	\newblock {\em Econometric Theory\/}~{\em 32\/}(3), 714–739.
	
	\bibitem[\protect\citeauthoryear{Jeong and Koo}{Jeong and Koo}{2025}]{JeongKoo2025_ISPRS}
	Jeong, M. and H.~Koo (2025).
	\newblock Evaluating spatio-temporal kriging with machine learning considering the sources of spatio-temporal variation.
	\newblock {\em ISPRS Int. J. of Geo-Information\/}~{\em 14\/}(6), 224.
	
	\bibitem[\protect\citeauthoryear{Khan, Almazah, EIlahi, Niaz, Al-Rezami, and Zaman}{Khan et~al.}{2023}]{khan2023spatial}
	Khan, M., M.~M. Almazah, A.~EIlahi, R.~Niaz, A.~Al-Rezami, and B.~Zaman (2023).
	\newblock Spatial interpolation of water quality index based on ordinary kriging and universal kriging.
	\newblock {\em Geomatics, Natural Hazards and Risk\/}~{\em 14\/}(1), 2190853.
	
	\bibitem[\protect\citeauthoryear{Kheir, Govind, Nangia, El-Maghraby, Elnashar, Ahmed, Aboelsoud, Gamal, and Feike}{Kheir et~al.}{2025}]{kheir2025hybridization}
	Kheir, A.~M., A.~Govind, V.~Nangia, M.~A. El-Maghraby, A.~Elnashar, M.~Ahmed, H.~Aboelsoud, R.~Gamal, and T.~Feike (2025).
	\newblock Hybridization of process-based models, remote sensing, and machine learning for enhanced spatial predictions of wheat yield and quality.
	\newblock {\em Computers and Electronics in Agriculture\/}~{\em 234}, 110317.
	
	\bibitem[\protect\citeauthoryear{Kiani, Motamedvaziri, Khaleghi, and Ahmadi}{Kiani et~al.}{2025}]{kiani2025spatial}
	Kiani, A., B.~Motamedvaziri, M.~R. Khaleghi, and H.~Ahmadi (2025).
	\newblock Spatial prediction of flood susceptible areas using machine learning methods in the siahkhor watershed of kermanshah province.
	\newblock {\em Earth Science Informatics\/}~{\em 18\/}(1), 20.
	
	\bibitem[\protect\citeauthoryear{LeSage and Pace}{LeSage and Pace}{2009}]{LeSage2009}
	LeSage, J. and R.~K. Pace (2009, January).
	\newblock {\em Introduction to Spatial Econometrics}.
	\newblock Chapman and Hall/CRC.
	
	\bibitem[\protect\citeauthoryear{Lovatto, Garc{\'\i}a~Arancibia, and Llop~Orzan}{Lovatto et~al.}{2022}]{Lovatto:thesis2022}
	Lovatto, M.~G., R.~Garc{\'\i}a~Arancibia, and P.~N. Llop~Orzan (2022).
	\newblock {Kriging semiparamétrico para datos univariados}.
	\newblock Master's thesis, Universidad Nacional de Rosario, Argentina.
	
	\bibitem[\protect\citeauthoryear{Meng, Yang, Yang, Zhang, De~Jesus, Correia, Fazeres-Ferradosa, Macek, Branco, and Zhu}{Meng et~al.}{2024}]{meng2024kriging}
	Meng, D., H.~Yang, S.~Yang, Y.~Zhang, A.~M. De~Jesus, J.~Correia, T.~Fazeres-Ferradosa, W.~Macek, R.~Branco, and S.-P. Zhu (2024).
	\newblock Kriging-assisted hybrid reliability design and optimization of offshore wind turbine support structure based on a portfolio allocation strategy.
	\newblock {\em Ocean Engineering\/}~{\em 295}, 116842.
	
	\bibitem[\protect\citeauthoryear{M{\'\i}nguez, Basile, and Durb{\'a}n}{M{\'\i}nguez et~al.}{2020}]{minguez2020alternative}
	M{\'\i}nguez, R., R.~Basile, and M.~Durb{\'a}n (2020).
	\newblock An alternative semiparametric model for spatial panel data.
	\newblock {\em Stat. Methods \& Applications\/}~{\em 29\/}(4), 669--708.
	
	\bibitem[\protect\citeauthoryear{Montero, M{\'\i}nguez, and Durb{\'a}n}{Montero et~al.}{2012}]{montero2012sar}
	Montero, J., R.~M{\'\i}nguez, and M.~Durb{\'a}n (2012).
	\newblock Sar models with nonparametric spatial trends. a p-spline approach.
	\newblock {\em Estad{\'\i}stica Espa{\~n}ola\/}~{\em 54\/}(177), 89--111.
	
	\bibitem[\protect\citeauthoryear{Oliver}{Oliver}{2003}]{Oliver2003}
	Oliver, D.~S. (2003, August).
	\newblock Gaussian cosimulation: Modelling of the cross-covariance.
	\newblock {\em Mathematical Geology\/}~{\em 35\/}(6), 681–698.
	
	\bibitem[\protect\citeauthoryear{Oodaira and ichi Yoshihara}{Oodaira and ichi Yoshihara}{1971}]{Oodaira1971}
	Oodaira, H. and K.~ichi Yoshihara (1971).
	\newblock The law of the iterated logarithm for stationary processes satisfying mixing conditions.
	\newblock {\em KODAI MATHEMATICAL SEMINAR REPORTS\/}~{\em 23\/}(3), 311--334.
	
	\bibitem[\protect\citeauthoryear{Pang, Wang, Lai, Zhang, Liang, and Song}{Pang et~al.}{2023}]{pang2023enhanced}
	Pang, Y., Y.~Wang, X.~Lai, S.~Zhang, P.~Liang, and X.~Song (2023).
	\newblock Enhanced kriging leave-one-out cross-validation in improving model estimation and optimization.
	\newblock {\em Computer Methods in Appl. Mechanics and Engineering\/}~{\em 414}, 116194.
	
	\bibitem[\protect\citeauthoryear{Robinson}{Robinson}{1988}]{robinson1988}
	Robinson, P.~M. (1988).
	\newblock Root-n-consistent semiparametric regression.
	\newblock {\em Econometrica\/}~{\em 56\/}(4), 931--954.
	
	\bibitem[\protect\citeauthoryear{Shen, LaRue, Fei, and Zhang}{Shen et~al.}{2024}]{Shen2024_HybridML_Geo}
	Shen, L., E.~LaRue, S.~Fei, and H.~Zhang (2024).
	\newblock Spatial prediction of plant invasion using a hybrid of machine learning and geostatistical method.
	\newblock {\em Ecology and Evolution\/}~{\em 14}, e11605.
	
	\bibitem[\protect\citeauthoryear{Stakhovych and Bijmolt}{Stakhovych and Bijmolt}{2009}]{stakhovych2009specification}
	Stakhovych, S. and T.~H. Bijmolt (2009).
	\newblock Specification of spatial models: A simulation study on weights matrices.
	\newblock {\em Papers in Regional Science\/}~{\em 88\/}(2), 389--409.
	
	\bibitem[\protect\citeauthoryear{Tadi{\'c}, Ili{\'c}, Ili{\'c}, Pavlovi{\'c}, and Tadi{\'c}}{Tadi{\'c} et~al.}{2024}]{Tadic2024_RemoteSensing_SIF}
	Tadi{\'c}, J.~M., V.~Ili{\'c}, S.~Ili{\'c}, M.~Pavlovi{\'c}, and V.~Tadi{\'c} (2024).
	\newblock Hybrid machine learning and geostatistical methods for gap filling and predicting solar‐induced fluorescence values.
	\newblock {\em Remote Sensing\/}~{\em 16\/}(10), 1707.
	
	\bibitem[\protect\citeauthoryear{Vagnini, Canal~Vieira, Longo, and Mura}{Vagnini et~al.}{2025}]{vagnini2025regional}
	Vagnini, C., L.~Canal~Vieira, M.~Longo, and M.~Mura (2025).
	\newblock Regional drivers of industrial decarbonisation: a spatial econometric analysis of 238 eu regions between 2008 and 2020.
	\newblock {\em Regional Studies\/}~{\em 59\/}(1), 2380369.
	
	\bibitem[\protect\citeauthoryear{Wadoux, Heuvelink, {de Bruin}, and Brus}{Wadoux et~al.}{2021}]{WADOUX2021}
	Wadoux, A. M.-C., G.~B. Heuvelink, S.~{de Bruin}, and D.~J. Brus (2021).
	\newblock Spatial cross-validation is not the right way to evaluate map accuracy.
	\newblock {\em Ecological Modelling\/}~{\em 457}, 109692.
	
	\bibitem[\protect\citeauthoryear{Walk}{Walk}{2010}]{Devroye2010}
	Walk, H. (2010).
	\newblock Strong laws of large numbers and nonparametric estimation.
	\newblock In L.~Devroye, B.~Karas{\"o}zen, M.~Kohler, and R.~Korn (Eds.), {\em Recent Developments in Appl. Probability and Stat.}, pp.\  1--40. Berlin: Springer.
	
\end{thebibliography}
\end{document}